\documentclass[journal,onecolumn]{IEEEtran}
\usepackage[T1]{fontenc}
\usepackage{bbm}
\usepackage{hyperref}
\usepackage{pgfplots}
\usepackage{tikz}

\usetikzlibrary{shapes,shadows,arrows}
\usepackage{mathtools,cuted}
\usepackage{latexsym}
\usepackage{blkarray}
\usepackage{pictex}
\usepackage{fancyvrb}
\usepackage{fancyhdr}
\usepackage{color}
\usepackage{xfrac}
\usepackage{subfig}
\usepackage{rawfonts}
\usepackage{graphics}
\usepackage{graphicx}
\usepackage{amssymb,amsmath,amsthm,amsfonts}
\usepackage{ifthen}
\usepackage{bbm}
\usepackage{cite}
\usepackage{enumerate}
\usepackage{algorithm}
\usepackage[]{algpseudocode}
\usepackage{verbatim}
\usepackage{balance}

\newtheorem{theorem}{Theorem}

\newtheorem{lemma}[theorem]{Lemma}

\newtheorem{remark}{Remark}
\newtheorem{definition}[theorem]{Definition}

\makeatletter
\def\BState{\State\hskip-\ALG@thistlm}
\makeatother

\begin{document}

\title{Secure Distributed Matrix Computation with Discrete Fourier Transform}

\author{
	\IEEEauthorblockN{Nitish Mital, Cong Ling and Deniz G\"{u}nd\"{u}z \\ 
	\IEEEauthorblockA{Department of Electrical \& Electronics Engineering, Imperial College London}\\
	Email: \{n.mital,c.ling,d.gunduz\}@imperial.ac.uk}
}

\maketitle

\begin{abstract}
We consider the problem of secure distributed matrix computation (SDMC), where a \textit{user} queries a function of data matrices generated at distributed \textit{source} nodes. We assume the availability of $N$ honest but curious computation servers, which are connected to the sources, the user, and each other through orthogonal and reliable communication links. Our goal is to minimize the amount of data that must be transmitted from the sources to the servers, called the \textit{upload cost}, while guaranteeing that no $T$ colluding servers can learn any information about the source matrices, and the user cannot learn any information beyond the computation result. We first focus on secure distributed matrix multiplication (SDMM), considering two matrices, and propose a novel polynomial coding scheme using the properties of finite field discrete Fourier transform, which achieves an upload cost significantly lower than the existing results in the literature. We then generalize the proposed scheme to include straggler mitigation, and to the multiplication of multiple matrices while keeping the input matrices, the intermediate computation results, as well as the final result secure against any $T$ colluding servers. We also consider a special case, called computation with own data, where the data matrices used for computation belong to the user. In this case, we drop the security requirement against the user, and show that the proposed scheme achieves the minimal upload cost. We then propose methods for performing other common matrix computations securely on distributed servers, including changing the parameters of secret sharing, matrix transpose, matrix exponentiation, solving a linear system, and matrix inversion, which are then used to show how arbitrary matrix polynomials can be computed securely on distributed servers using the proposed procedure.
\end{abstract}
\let\thefootnote\relax\footnotetext{This work was supported in part by the European Union`s H2020 research and innovation programme under the Marie Sklodowska-Curie Action SCAVENGE (grant agreement no. 675891), and by the European Research Council (ERC) Starting Grant BEACON (grant agreement no. 677854).}

\section{Introduction}

In the era of big data, performing computationally intensive operations locally on a single machine is infeasible, and clients often rely on powerful cloud servers to carry out demanding computation tasks. In the so-called \textit{serverless computing paradigm}, clients can request computationally expensive tasks to be performed on massive datasets, potentially generated at multiple geographically distributed locations, using special purpose computing servers (eg., Amazon Web Services (AWS), Microsoft Azure, Google Cloud). 
While serverless computing provides significant flexibility and speed up, it also leads to growing data privacy concerns, as the corporations that provide computation services also provide many other digital services, and have access to unprecedented amounts of private user data. 
Therefore, algorithms that would allow users to benefit from powerful untrustworthy servers while keeping their data private are of significant interest. 

Our goal in this paper is to design efficient secure distributed matrix computation (SDMC) algorithms, which keep data private from the potentially colluding computing servers as well as the entities requesting the computations.

We consider $\Gamma \geq 1$ data \textit{sources}, represented as matrices $\mathbf{A}^{(1)}, \ldots, \mathbf{A}^{(\Gamma)}$ on an appropriate finite field. A \textit{user} desires to compute a function of these matrices, $G(\mathbf{A}^{(1)}, \ldots, \mathbf{A}^{(\Gamma)})$, with the help of $N$ computing servers. The servers are connected to the sources and to each other with orthogonal and reliable links. Similarly, computations carried out by the servers are conveyed to the user over orthogonal and reliable communication links. For a given number of $N$ servers, our goal will be to minimize the amount of data that must be uploaded from the sources to the servers, which we refer to as the \textit{upload cost}. The upload cost often determines the financial cost of serverless computing, but minimizing it would also reduce the overall computing time as it limits the amount of computations that must be carried out by the servers, as well as the communication latency from the sources to the servers, which may be prohibitive especially when the data sources are geographically distant from the servers. For example, the source nodes may be distant hospitals sharing medical data of patients, and the user may be a research institute or a pharmaceutical company making certain queries on the data. In addition to correct computation of the request, we also want to guarantee the privacy of the input data against the servers as well as the requesting user. We impose information theoretic perfect privacy guarantees such that any $T$ colluding servers must not learn anything about the data sources, or the user must not learn anything about the data sources apart from the computation result. We assume that all the servers are honest and responsive, but curious, which means that they follow the prescribed protocol honestly, but any $T$ of them may collude to try to deduce information about the input matrices. We will also consider the special setting of \textit{computation with own data}, in which case the user wants to compute a function on its own data matrices using the available computing servers. In this case we drop the privacy requirement against the user, and the problem lends itself to further optimization. 

We will first focus on the secure distributed matrix multiplication (SDMM) problem, which has received significant recent interest. Large scale matrix multiplication is a fundamental building block of matrix computations in many machine learning, optimization, and signal processing algorithms. It is also one of the most computationally intensive operations. Moreover, it can be easily distributed across multiple servers thanks to its inherently parallel structure. We will first consider the multiplication of $\Gamma =2$ matrices, which will allow us to introduce the main ideas behind our design. We then extend our analysis to the multiplication of multiple matrices, as well as to other fundamental matrix operations, which, when combined with matrix multiplication, allow computation of arbitrary polynomials of matrices.

\subsection{Related Work}

The cryptography community has extensively studied the problem of secure multi-party computation (MPC), also known as secure function evaluation, in which Alice and Bob, having inputs $x$ and $y$, respectively, want to compute a function $f(x,y)$ jointly, without any of them learning anything about the other's input either from the communication, or from the result of the computation \cite{5394944}. The SDMM problem is related to MPC yet different; the design has to ensure that no computing server learns anything about the original data, but we can decide which part of the data is revealed to each server and in what form. Fully homomorphic encryption (FHE) is a class of cryptographic schemes that allow computations on ciphertexts, generating an encrypted result which, when decrypted, matches the result of the operations as if they had been performed on plaintext. These techniques rely on working over polynomial rings, and their security is based on the assumed (or proven) hardness of problems in ideal lattices \cite{10.1007/978-3-319-19962-7_27,EfficientSecureMatrixMultiplicationOverLWEBasedHomomorphicEncryption}. However, existing FHE schemes are slow and impractical. ``Somewhat homomorphic encryption'' (SHE) has been proposed as an alternative, which allows a limited number of homomorphic operations on ciphertexts. SHE is relatively faster, and ciphertext packing methods have been proposed for operations like secure inner products \cite{10.1007/978-3-319-19962-7_27,10.1007/978-3-642-54568-9_3}, and secure matrix multiplications \cite{EfficientSecureMatrixMultiplicationOverLWEBasedHomomorphicEncryption}, which generalizes Yasuda et al.'s packing method for inner products in \cite{10.1007/978-3-642-54568-9_3}. Some works also propose methods for performing other matrix computations in an information theoretically secure manner, like Gaussian elimination, matrix inversion, comparison, equality test, or exponentiation \cite{Bouman2018NewPF,10.1007/3-540-44647-8_7}.

There is also a growing literature on distributed matrix multiplication, where a lot of effort has been put into speeding up computations, increasing reliability, and/or reducing communication overhead using coding and communication theoretic ideas \cite{8002642,pmlr-v70-tandon17a,NIPS2017_7027,8765375,Dutta2018AUC,8437563,DBLP:journals/corr/abs-1806-00939}. The initial papers considered a slightly different context of speeding up parallel computations by introducing ``computation redundancy'' to mitigate the problem of straggling servers \cite{8002642,pmlr-v70-tandon17a}. Straggling servers refer to slow/unresponsive servers due to which completion of the computation is delayed. A standard way of dealing with stragglers is to introduce ``computation redundancy'', that is, assigning extra computations to each server. In coded computation against stragglers, the performance metric is the \textit{recovery threshold}; that is, the minimum integer $r$ such that the computation result is recoverable from any $r$ successful (non-delayed, non-faulty) servers. The common theme in the `coded computation' literature is to treat stragglers as `erasures' in communications, and exploit ideas for coding against erasures, which allows reliable reconstruction of the desired result from an arbitrary set of successfully received symbols. Reference \cite{NIPS2017_7027} uses polynomial codes to construct a scheme in which the computation is completed as long as any $K$ out of $N$ evaluations of a polynomial are received from the servers. To multiply two matrices $\mathbf{A}$ and $\mathbf{B}$ with the help of $N$ servers, the polynomial code in \cite{NIPS2017_7027} partitions $\mathbf{A}$ row-wise and $\mathbf{B}$ column-wise (row-by-column partitioning), and generates encoded matrices as evaluations of a polynomial with the blocks as the coefficients, similarly to Reed Solomon codes \cite{doi:10.1137/0108018}. A follow-up work \cite{8765375} proposes a new polynomial coded computation scheme called MatDot, which achieves the optimal recovery threshold for column-wise partitioning of matrix $\mathbf{A}$ and row-wise partitioning of $\mathbf{B}$ (sum-of-outer-products method), which we shall refer to as the column-by-row partitioning henceforth in the paper, at the expense of an increase in the communication cost. In \cite{8765375}, the authors also propose PolyDot codes that interpolate between the polynomial codes of \cite{NIPS2017_7027} and MatDot codes. PolyDot codes are later improved as generalized PolyDot (GPD) codes in \cite{Dutta2018AUC}, which achieve the optimal recovery threshold for any arbitrary partitioning of the input matrices. Entangled polynomial codes, proposed in parallel in \cite{8437563}, also achieve the same performance as GPD codes. Bivariate polynomial codes are introduced in \cite{Hasircioglu2020BivariatePC} for straggler mitigation when servers can compute and transmit multiple partial computations. 

Subsequent papers, inspired by the works on straggler mitigation in distributed matrix multiplication, consider the SDMM problem from an information theoretic perspective. These papers aim for information theoretic security, independent of the computational capacities of the attackers, as opposed to cryptographic techniques. These works typically assume that the data to be used for computations belongs to the user; hence, they focus on privacy against the servers. The earlier papers on SDMM consider \textit{download rate} as the performance metric, which is defined as the ratio of the number of bits required to represent the computation result to the total number of bits that the servers must transmit to the user. Reference \cite{2018arXiv180600469C} uses the idea of polynomial codes from the literature on straggler mitigation to propose an SDMM scheme based on Shamir's secret sharing scheme \cite{Shamir:1979:SS:359168.359176}, which is shown to achieve the optimal download rate for \textit{one-sided} SDMM (where only one of the matrices is kept secure). An achievable scheme is also proposed for \textit{two-sided} SDMM (both matrices are kept secure). Reference \cite{8849446} introduces GASP codes, which improve the download rate for two-sided SDMM by aligning the degrees of the terms in the polynomial code so that the desired products appear as distinct terms.

In \cite{2019arXiv190806957J}, new converse bounds on the optimal download rate for SDMM are obtained by showing that the capacity of a multi-message X-secure T-private information retrieval (MM-XSTPIR) problem (\cite{8341744,8598994}) provides an upper bound on the download rate of the SDMM problem. The optimal download rate of the MM-XSTPIR problem is shown to depend on the dimensions of the matrices $\mathbf{A}$ and $\mathbf{B}$. The scheme in \cite{2019arXiv190806957J} allows the joint retrieval of a batch of matrix products (\textit{batch multiplication}), instead of multiplying two matrices using the \textit{matrix partitioning} approach, resulting in a coding gain. Other recent works on secure distributed batch matrix multiplication include \cite{DBLP:journals/corr/abs-2001-05101} and \cite{Chen_2020}.

Reference \cite{DBLP:journals/corr/abs-1806-00939} combines straggler mitigation and secure computation using the batch multiplication approach. Byzantine security is also considered, which refers to security against adversarial servers that may actively corrupt the results they send back to the user. Their scheme is based on Lagrange polynomials, and it achieves the optimal recovery threshold for any multi-linear function computation. Since the scheme in \cite{DBLP:journals/corr/abs-1806-00939} is designed for batch computation of any function, we can adapt it for matrix multiplication using a matrix-partitioning based approach, and compare its performance with that of other matrix partitioning based schemes. This can be done by treating the partitions of the matrices to be multiplied as batches of data. Therefore, if $\mathbf{A}$ is partitioned column-wise and $\mathbf{B}$ is partitioned row-wise into $K$ partitions each, the scheme in \cite{DBLP:journals/corr/abs-1806-00939} has the recovery threshold of $2(K+T-1)+S+2A+1$, where $S$ and $ A$ are the numbers of stragglers and Byzantine adversaries, respectively. If $\mathbf{A}$ is partitioned row-wise and $\mathbf{B}$ is partitioned column-wise, into $K$ and $L$ partitions respectively, the recovery threshold is $2(KL+T-1) + S+2A+1$, which is the same performance as that of the GASP codes in \cite{8849446} for big $T$, which can be seen by setting $S=A=0$. For arbitrary matrix partitions, the Lagrange coded scheme does not perform as well as the secure generalized PolyDot (SGPD) codes introduced later in \cite{Aliasgari:ISIT:19}. In \cite{Aliasgari:ISIT:19}, the trade-off between the download rate and the recovery threshold, first studied in \cite{2018arXiv180110292D} using the MatDot and PolyDot schemes for straggler mitigation, is extended to SDMM. The SGPD codes achieve the same recovery threshold as Lagrange codes for the sum-of-outer-products method.

In \cite{Kakar2019UplinkDownlinkTI}, the trade-off between the upload and download costs for SDMM is studied. While the \textit{download cost} is simply the reciprocal of the download rate, the \textit{upload cost} is defined as the ratio of the total number of bits that the user must send to the servers to the total size in bits of the data matrices. In \cite{Kakar2019UplinkDownlinkTI}, a secure cross subspace alignment (SCSA) scheme with adjustable upload cost (USCSA) is presented for SDMM. The tradeoff between the upload and download costs for SDMM schemes using only the row-by-column partitioning is studied in \cite{Kakar2019UplinkDownlinkTI}, but not the tradeoff for SDMM schemes using the column-by-row partitioning.

The work that is most related to ours is \cite{8437651}, which considers the setting in which the data is generated at distributed source nodes, and does not belong to the user requesting the computation. It extends the BGW (Ben-Or, Goldwasser and Widgerson) scheme from \cite{10.1145/3335741.3335756}, which was first proposed in the context of secure MPC, for multiplication of matrices using a connected network of computing servers. The sources are assumed not to be connected with each other, while the servers are. The latter assumption is exploited to reduce the communication to the user by allowing the servers to cooperate securely. Thus, the servers share their results from the first round of computation among each other using Shamir's secret sharing scheme, and compute a linear combination of the received shares of the results. This inter-server cooperation allows a smaller number of servers to send these linear combinations to the user. This particular model also imposes privacy against the user, that is, the user cannot learn anything about the data beyond what it learns from the computation result. This constraint is not imposed in other papers, where the user is the source of the input data; and neither in \cite{2019arXiv190806957J}, where distributed source nodes generate the data.

\subsection{Main contributions:}
With respect to the rich literature on the topic that we have summarized above, the main novel contributions of our work can be summarized as follows: 
\begin{itemize}
    \item We first introduce a novel polynomial coding scheme exploiting the properties of discrete Fourier transform, and show that it achieves a near optimal upload cost for SDMM of two matrices, while achieving the optimal upload cost for the special case of computation with own data, in which the user has access to the matrices used in the computations. The key difference between the existing polynomial coded SDMM schemes and the one proposed in this paper is that instead of evaluating the polynomial over some arbitrary distinct points, here we evaluate the polynomial at the $N$-th roots of unity, which yields a discrete Fourier transform over the finite field, also known as the number theoretic transform (NTT).
    \item The proposed scheme can be implemented in an efficient manner using the recently developed fast Fourier transform (FFT) algorithm on finite fields \cite{7565465}, and has negligible decoding complexity. 
    \item  We generalize the proposed scheme for SDMM of two matrices to introduce straggler mitigation.
    \item We extend the proposed scheme to securely multiply multiple matrices on distributed servers securely. Our scheme has a significantly lower upload cost than the existing alternatives, and is naturally scalable to the multiplication of an arbitrary number of matrices.
    \item  We also present some schemes for other matrix operations, such as addition, transpose, exponentiation, changing the parameters of the secret shares, and solving linear systems, which includes computing the matrix inverse. As a result, it is shown that arbitrary matrix polynomials can be computed securely on distributed servers.
\end{itemize}

\textbf{Notations:} The notation $[b]$ denotes the set of consecutive integers $\{1,\ldots,b \}$. The set of natural numbers is denoted by $\mathbb{N}$. Sets are denoted by calligraphic letters $\mathcal{L}$. Matrices and vectors are denoted by bold upper-case letters $\mathbf{A}$ and bold lower-case letters $\mathbf{a}$, respectively. The notation $[[ \mathbf{A} ]]_i $ denotes the secret share of matrix $\mathbf{A}$ delivered to server $i$, while the notation $[[ \mathbf{A} ]]_{\mathcal{L}}$ denotes the set of secret shares of matrix $\mathbf{A}$ delivered to servers belonging to the set $\mathcal{L}$.



\section{System Model}
We consider $\Gamma$ source nodes, $N\geq 2$ servers, and one user, for some $\Gamma, N \in \mathbb{N}$ (see Fig. \ref{fig:scheme}). Each source node is connected to each server through an orthogonal link. Each pair of servers is connected to each other, and each server is connected to the user through a private link. Each source node $\gamma \in [\Gamma] \triangleq \{1,\ldots, \Gamma \}$ has access to an input matrix $\mathbf{A}^{(\gamma)} \in \mathbb{F}_q^{m_{\gamma} \times m'_{\gamma}}$, for $m_{\gamma},m'_{\gamma} \in \mathbb{N}, \gamma \in [\Gamma]$, and a finite field $\mathbb{F}_q$ with $q$ elements. Using $N$ computing servers, the user wants to securely compute the result of a function $\mathbf{C}=G(\mathbf{A}^{(1)},\ldots, \mathbf{A}^{(\Gamma)})$, where $G$ is an arbitrary polynomial function, assuming appropriate matrix dimensions. We assume that the entries of $\mathbf{A}^{(\gamma)}, \gamma \in [\Gamma]$, are independent of each other and uniformly distributed over $\mathbb{F}_q$. We also assume that the servers are honest, but curious, which means that each server honestly follows the protocol without spurious insertions, yet may infer information about the inputs passively. 
Similarly to \cite{8437651}, the system operates in three phases: (1) Sharing, (2) Computation and communication, and (3) Reconstruction. A detailed description of these phases is as follows.
\begin{description}
\item[1)]\textbf{Sharing phase:} In this phase, the source $\gamma $ sends secret shares of matrix $\mathbf{A}^{(\gamma)}$, denoted by $[[\mathbf{A}^{(\gamma)}]]_i$, to server $i$. $[[\mathbf{A}^{(\gamma)}]]_i$ is a function of input matrix $\mathbf{A}^{(\gamma)}$ and a secret key $\mathbf{Sk}^{(\gamma)}$, and its dimensions will depend on the computing scheme.
\item[2)]\textbf{Computation and communication:} In this phase, the servers process the data they have received from the sources, and may also exchange messages with each other. We denote the set of all messages that server $i$ sends to server $i'$ in this phase by $\mathcal{M}_{i,i'}$.
\item[3)]\textbf{Reconstruction:} In this phase, every server $i\in [N]$ sends a message $[[\mathbf{C}]]_i$ to the user, who decodes the received messages to recover the desired computation.
\end{description}

The scheme must satisfy the following four constraints.

\paragraph{Correctness}
The user must be able to decode the final function $\mathbf{C}=G(\mathbf{A}^{(1)},\ldots, \mathbf{A}^{(\Gamma)})$ from the responses it receives from the servers, $[[\mathbf{C}]]_1, \ldots, [[\mathbf{C}]]_N$. The correctness constraint is imposed by:
\begin{align}
    H\left(\mathbf{C}\big| [[\mathbf{C}]]_1, \ldots, [[\mathbf{C}]]_N\right)=0.
\end{align}
\begin{remark}
We will also consider the special case of \textit{computation with own data}, where the data matrices used for computation belong to the user, as in the distributed computation scenario studied in \cite{2018arXiv180600469C,8849446,DBLP:journals/corr/abs-1806-00939,Aliasgari:ISIT:19}. In this scenario, the correctness constraint becomes:
\begin{align}
    H\left(\mathbf{C}\big| [[\mathbf{C}]]_1, \ldots, [[\mathbf{C}]]_N, \mathbf{Sk}^{(1)},\ldots, \mathbf{Sk}^{(\Gamma)}\right)=0.
\end{align}
\end{remark}

\paragraph{Security against colluding servers}
The goal is to recover $\mathbf{C}$ reliably and securely even if any $T<N$ servers collude to extract some information about the input matrices. Hence, for any $\mathcal{L} \subset[N]$ with $\vert \mathcal{L} \vert \leq T $, the encoded matrices $[[\mathbf{A}^{(\gamma)}]]_{\mathcal{L}}$, and the set of messages communicated by the servers in $\mathcal{L}^c$ to the servers in $\mathcal{L}$, denoted by $\mathcal{M}_{\mathcal{L}^c,\mathcal{L}}$, must not reveal any information about the source matrices, $\{ \mathbf{A}^{(\gamma)} \}_{\gamma=1}^{\Gamma}$. Accordingly, the security constraint is specified as,
\begin{align} 
   I\left(\left\{\mathbf{A}^{(\gamma)}\right\}_{\gamma=1}^{\Gamma}; \left\{[[\mathbf{A}^{(\gamma)}]]_{\mathcal{L}}\right\}_{\gamma=1}^{\Gamma},\mathcal{M}_{\mathcal{L}^c,\mathcal{L}}\right)=0,\hspace{0.6cm} \forall \mathcal{L} \subseteq [N], \vert \mathcal{L} \vert \leq T.\label{security_constraint}
\end{align}
 
\paragraph{Security against the user}
The user must not gain additional information about the input matrices beyond the result of the function $G$. This is defined as:
\begin{align}
   \nonumber  I\left(\mathbf{A}^{(1)},\ldots,\mathbf{A}^{(\Gamma)} ; [[\mathbf{C}]]_1,\ldots, [[\mathbf{C}]]_N \big| \mathbf{C}\right)=0.
\end{align}

For a given number of servers $N$ and security requirement $T$, the performance will be measured in terms of the \textit{upload cost} from the source nodes to the servers. The upload cost is defined as follows:
\begin{align}\label{storage_cost}
    \chi_{UL}\triangleq \frac{\sum_{\gamma=1}^{\Gamma}\sum_{i=1}^{N} H([[\mathbf{A}^{(\gamma)}]]_i)}{\sum_{\gamma=1}^{\Gamma} H(\mathbf{A}^{(\gamma)})},
\end{align}
and it quantifies the normalized amount of information that must be delivered to the servers. In most cases the amount of computation that must be carried out by each server depends on the amount of information delivered to it, and many cloud computing services charge users based on the amount of information delivered and stored at each server. Hence, minimizing the upload cost will reduce both the latency and the cost of computations. Accordingly, our objective is to securely compute $\mathbf{C}$ by incurring the minimum upload cost. Due to the symmetry across the servers, we will assume (unless stated otherwise) that the secret shares sent from each source to each server is of the same size, and is a $(\sfrac{1}{K})^{th}, K\in \mathbb{N}$, fraction of the size of the input matrices.

\begin{definition}
An $(N,K,T)$ SDMC scheme uses $N$ servers, sends $(\sfrac{1}{K})^{th}$ fraction of the input data matrices' size to each server, and is secure against any $T$ colluding servers. Hence, the upload cost of an $(N,K,T)$ SDMC scheme is given by $\chi_{UL} = \frac{N}{K}$.
\end{definition}

\section{SDMM for $\Gamma=2$}\label{polymat_eg}
In this section we focus exclusively on the secure distributed multiplication of two matrices, i.e., $\Gamma = 2$. This problem will allow us to present the main ideas behind our coded computation scheme. We first present our main result in the following theorem.
\begin{theorem}
We can securely multiply two matrices using $N$ servers, $T$ of which may collude, with $N> 2T$, with an upload cost of $\frac{N}{N-2T}$. In other words, an $(N,N-2T,T)$ SDMM scheme for two input matrices is achievable.
\end{theorem}

We present the following example to illustrate the essential ingredients of the scheme. For ease of notation, we use the notation $\mathbf{A}^{(1)}=\mathbf{A}$ and $\mathbf{A}^{(2)}=\mathbf{B}$ in this section.

\textbf{Example:}
We consider distributed multiplication of matrices $\mathbf{A} \in \mathbb{F}_q^{m\times n}$ and $\mathbf{B} \in \mathbb{F}_q^{n\times p}$ over $N=7$ servers, any $T=2$ of which may collude (see Fig. \ref{fig:scheme}).

\subsubsection{Sharing phase}
The matrices $\mathbf{A}$ and $\mathbf{B}$ are partitioned into $K=N-2T=3$ blocks of dimensions $m\times \frac{n}{3}$ and $\frac{n}{3}\times p$, respectively, as follows:
\begin{align}\label{partition}
    \mathbf{A}=\left[ \begin{array}{ccc}
         \mathbf{A}_1  & \mathbf{A}_2 & \mathbf{A}_3
    \end{array} \right] , \mathbf{B}=\left[ \begin{array}{c}
         \mathbf{B}_1  \\
         \mathbf{B}_2 \\
         \mathbf{B}_3
    \end{array} \right] ,
\end{align}
and the product $\mathbf{C}=\mathbf{AB}$ is given by
\begin{align}
    \mathbf{C}=\mathbf{A}_1\mathbf{B}_1 + \mathbf{A}_2\mathbf{B}_2 + \mathbf{A}_3\mathbf{B}_3.
\end{align}

\begin{figure}
    \centering
    \includegraphics[scale=0.74]{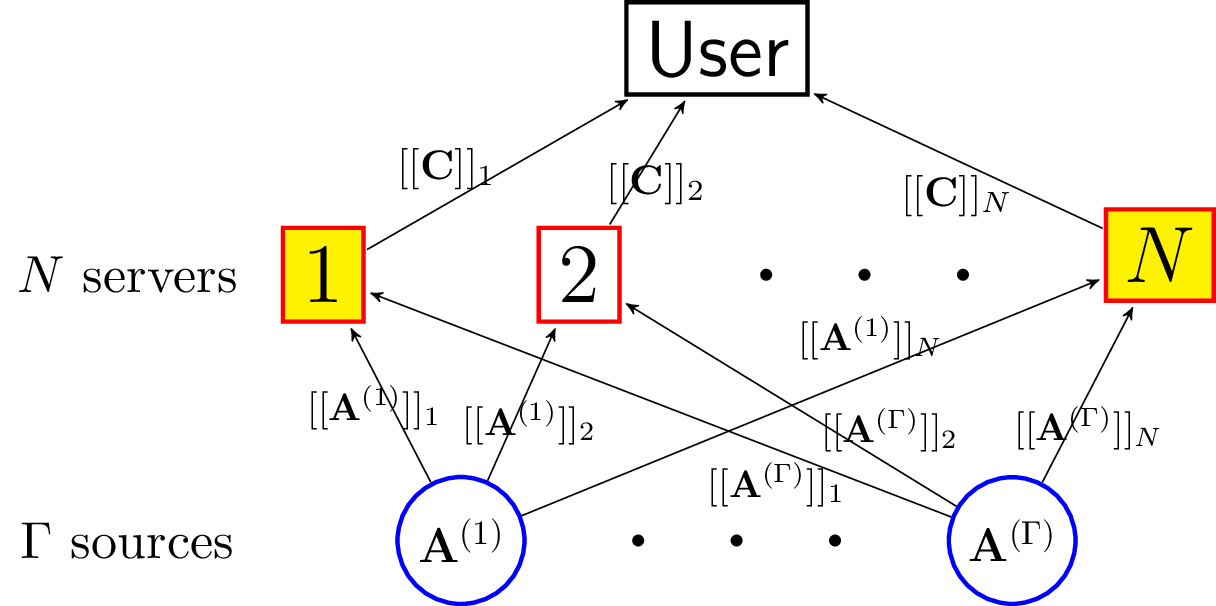}
    \caption{System model for $\Gamma$ matrices with $N$ servers, any $T$ of which may collude (shown as shaded servers).}
    \label{fig:scheme}
\end{figure}

The matrices $\mathbf{R}_i\in \mathbb{F}_q^{m \times \frac{n}{K}}$ and $\mathbf{S}_i\in \mathbb{F}_q^{\frac{n}{K} \times p}, i \in [T]$, are generated, whose entries are independent and identically distributed (i.i.d.) uniform random variables from the finite field $\mathbb{F}_q$. The matrices $\mathbf{R}=[\mathbf{R}_1,\ldots,\mathbf{R}_T]$ and $\mathbf{S}=[\mathbf{S}_1,\ldots,\mathbf{S}_T]$ are used as the secret keys to encode the input matrices $\mathbf{A}$ and $\mathbf{B}$, respectively. The following two polynomials are constructed to encode the matrices:
\begin{align}
    \mathbf{A}(x)&=\mathbf{A}_1 + \mathbf{A}_2 x + \mathbf{A}_3 x^2 + \mathbf{R}_1 x^3 + \mathbf{R}_2 x^4 \label{poly_A}\\
    \mathbf{B}(x)&= \mathbf{B}_1 + \mathbf{B}_2 x^{-1} + \mathbf{B}_3 x^{-2} + \mathbf{S}_1 x^{-5} + \mathbf{S}_2 x^{-6}. \label{poly_B}
\end{align}
Note that the encoding polynomials of $\mathbf{A}$ and $\mathbf{B}$ are different from each other, depending on the order of multiplication. We refer to the encoding of matrix $\mathbf{A}$ as ``left-encoding'', and that of matrix $\mathbf{B}$ as ``right-encoding''. We refer to each of the $N$ secret shares obtained from left-encoding of matrix $\mathbf{A}$ as an ``$(N,K,T)$ left-share'', where the $i^{th}, \forall i\in [N]$ secret share is denoted by $[[\mathbf{A}]]_i^{L}$, and each of those from right-encoding of matrix $\mathbf{B}$ as an ``$(N,K,T)$ right-share'', where the $i^{th}, \forall i \in [N]$ secret share is denoted by $[[\mathbf{B}]]_i^{R}$.

Let $\alpha_7$ be a primitive $7^{th}$ root of unity in $\mathbb{F}_q$, and $1,\alpha_7,\alpha_7^2,\ldots, \alpha_7^6$ be the $7^{th}$ roots of unity in $\mathbb{F}_q$. Then the polynomials $\mathbf{A}(x)$ and $\mathbf{B}(x)$ are evaluated at $\alpha_7^{i-1}$ to obtain $[[\mathbf{A}]]^L_i,[[\mathbf{B}]]^R_i$; that is, $[[\mathbf{A}]]^L_i=\mathbf{A}(\alpha_7^{i-1})$ and $[[\mathbf{B}]]^R_i=\mathbf{B}(\alpha_7^{i-1}), i=1,\ldots,7$, and the secret shares $[[\mathbf{A}]]^L_i,[[\mathbf{B}]]^R_i$ are sent to server $i$, $i=1,\ldots, 7$. The number of symbols sent by the source nodes to the servers is given by $H([[\mathbf{A}]]_i)+H([[\mathbf{B}]]_i)=\frac{mn}{K}+\frac{np}{K}$. \begin{remark}
We note here that evaluating the polynomials $\mathbf{A}(x)$ and $\mathbf{B}(x)$ at the roots of unity is equivalent to computing the discrete Fourier transform of the sequences $\left\{\mathbf{A}_1,\ldots,\mathbf{A}_K,\mathbf{R}_1,\ldots, \mathbf{R}_T\right\}$ and $\left\{\mathbf{B}_1,\ldots,\mathbf{B}_K,\mathbf{0}_1,\ldots,\mathbf{0}_T, \mathbf{S}_1,\ldots, \mathbf{S}_T\right\}$ in a finite field, where $\mathbf{0}_k, k=1,\ldots, T,$ are zero matrices. This can be carried out efficiently using the FFT algorithm for finite fields introduced in \cite{7565465}.
\end{remark}

\subsubsection{Computation phase}
Server $i$ computes the product
\begin{align}
    [[\mathbf{C}]]_i= [[\mathbf{AB}]]_i = [[\mathbf{A}]]^L_i [[\mathbf{B}]]^R_i,
\end{align}
which is equivalent to the evaluation of the polynomial 
\begin{align}
   \mathbf{C}(x)= & \mathbf{A}(x) \mathbf{B}(x)\\
    =& \mathbf{A}_1\mathbf{S}_2 x^{-6} + \left(\mathbf{A}_1\mathbf{S}_1 + \mathbf{A}_2\mathbf{S}_2 \right)x^{-5} + \cdots 
   + \left( \mathbf{A}_1\mathbf{B}_2 + \mathbf{A}_2\mathbf{B}_3 + \mathbf{R}_2\mathbf{S}_1 \right) x^{-1}  + \left(\sum_{l=1}^{3} \mathbf{A}_l\mathbf{B}_l \right) + \cdots + \mathbf{R}_2\mathbf{B}_1 x^{4} \label{poly_C} \\
&= \sum_{l=1}^3  \mathbf{A}_l\mathbf{B}_l + \left( \mathbf{A}_2\mathbf{B}_1 + \mathbf{A}_1\mathbf{S}_2 + \mathbf{A}_3\mathbf{B}_2 + \mathbf{R}_1\mathbf{B}_3 \right) x + \cdots + \left( \mathbf{A}_3\mathbf{S}_1 + \mathbf{R}_1\mathbf{S}_2 + \mathbf{R}_2\mathbf{B}_1 \right)x^4 + \nonumber\\
& \hspace{5cm} \cdots + \left( \mathbf{A}_1\mathbf{B}_2 + \mathbf{A}_2\mathbf{B}_3 + \mathbf{R}_2\mathbf{S}_1 \right)x^6  \label{aliased_polyC} 
\end{align}
at $x=\alpha_7^{i-1}$.

\begin{remark}
If we evaluate the polynomial over some arbitrary distinct points, then the product of two secret shares corresponds to the product of two polynomials, which is equivalent to the so-called \textit{linear convolution}. On the other hand, if we evaluate the polynomial over the $N$-th roots of unity, as is done in our paper, then the multiplication of two secret shares corresponds to the product of two polynomials modulo $x^N-1$, which is equivalent to the so-called \textit{circular convolution}. While the circular convolution introduces aliasing, since the only desired term for inner-product type of matrix multiplication is the constant term (i.e., the DC term), the correctness of the result is guaranteed despite aliasing. Such aliasing allows further interference alignment, which improves the efficiency compared to linear convolution. 
\end{remark}

\subsubsection{Reconstruction phase}
The servers send $[[\mathbf{C}]]_i$'s to the user. We know that 
\begin{align}\label{zero_sum}
    \sum_{i=1}^{N} (\alpha_N^{i-1})^s=0, \hspace{0.5cm}  \forall s: N \nmid s.
\end{align}

Therefore, the user computes the average of the received responses to obtain the final result:
\begin{align}
  \frac{1}{7}  \sum_{i=1}^{7} [[\mathbf{C}]]_i =\sum_{l=1}^{3} \mathbf{A}_l\mathbf{B}_l  , \label{final_sum}
\end{align}
because the non-constant terms from Eq. (\ref{poly_C}) sum to $0$ thanks to Eq. \eqref{zero_sum}.

For the general case with $N$ servers, $T$ of which can collude, the input matrices are partitioned into $K=N-2T$ submatrices similarly to Eq. \eqref{partition}. Then the product can be written as $\mathbf{C}=\mathbf{AB}=\sum_{l=1}^{K} \mathbf{A}_l\mathbf{B}_l$.

The matrices $\mathbf{A}$ and $\mathbf{B}$ are encoded with the following two polynomials:
\begin{align}
    \mathbf{A}(x)&= \sum_{l=1}^{K} \mathbf{A}_{l}x^{l-1} + \sum_{l=1}^{T}\mathbf{R}_{l} x^{K+l-1}, \label{polyA}
\end{align}
and
\begin{align}
    \mathbf{B}(x)&= \sum_{l=1}^{K} \mathbf{B}_{l}x^{-l+1} + \sum_{l=1}^{T}\mathbf{S}_{l} x^{-K-T-l+1}. \label{polyB}
\end{align}

We define the product polynomial as follows.
\begin{align}
    \mathbf{C}(x)=\sum_{l=1}^{K} \mathbf{A}_l \mathbf{B}_l  + (\text{non-constant terms}). \label{polyC}
\end{align}

The goal is to recover the constant term in $\mathbf{C}(x)$ from the computations of $N$ servers. The polynomials $\mathbf{A}(x)$ and $\mathbf{B}(x)$ are evaluated at the $N^{th}$ roots of unity, denoted by $1,\alpha_N, \alpha_N^2, \ldots, \alpha_N^{N-1} \in \mathbb{F}_q$, where $\alpha_N$ is a primitive $N^{th}$ root of unity in $\mathbb{F}_q$. Thus, the values sent to server $i$ are $[[\mathbf{A}]]^L_i=\mathbf{A}(\alpha_N^{i-1})$ and $[[\mathbf{B}]]^R_i=\mathbf{B}(\alpha_N^{i-1})$.  Server $i\in [N]$ computes the share $[[\mathbf{C}]]_i=[[\mathbf{A}]]^L_i[[\mathbf{B}]]^R_i=\mathbf{C}(\alpha_N^{i-1})$, and sends it to the user. 

Thanks to Eq. \eqref{zero_sum}, the user obtains the desired result by averaging the received $[[\mathbf{C}]]_i$'s.
\begin{align}\label{constant_term}
  \frac{1}{N}  \sum_{i=1}^{N} [[\mathbf{C}]]_i = \sum_{l=1}^{K} \mathbf{A}_l \mathbf{B}_l.
\end{align}

\begin{remark}
Since the polynomials must be evaluated at the $N^{th}$ roots of unity in the proposed scheme, the finite field must be chosen to guarantee the presence of all the $N^{th}$ order roots of unity. Therefore, we must have $N \vert (q-1)$. This can be satisfied by appropriately choosing the field size $q$. This also guarantees that the multiplicative inverse of $N$ exists in $\mathbb{F}_q$, so that the $\frac{1}{N}$ in Eq. \eqref{constant_term} exists.
\end{remark}

\subsection{Special case: Computation with own data}\label{special_case}
For the special case in which the data matrices belong to the user, the upload cost can be further reduced since the user also has access to the random secret keys used for generating the secret shares. Most previous literature on SDMM considers this special case. The user partitions the input matrices into $K$ blocks as in Eq. \eqref{partition}, where $K=N-T$. The user then encodes the matrices using the following polynomials.
\begin{align}
    \mathbf{A}(x)&= \sum_{l=1}^{K} \mathbf{A}_{l}x^{l-1} + \sum_{l=1}^{T}\mathbf{R}_{l} x^{K+l-1}, \label{polyA_same}
\end{align}
and
\begin{align}
    \mathbf{B}(x)&= \sum_{l=1}^{K} \mathbf{B}_{l}x^{-l+1} + \sum_{l=1}^{T}\mathbf{S}_{l} x^{-K-l+1}, \label{polyB_same}
\end{align}
where $\mathbf{A}(x)$ is the same as in Eq. \eqref{polyA} while we have a slight change from Eq. \eqref{polyB} in the way the secret key is embedded into the $\mathbf{B}(x)$ polynomial. The user evaluates the polynomials on the $N^{th}$ roots of unity to generate the secret shares, and sends these shares to the servers. The servers return their computed results back to the user, where the constant term in the product polynomial $\mathbf{C}(x)$ is now given by $\sum_{i=1}^{K}\mathbf{A}_i\mathbf{B}_i + \sum_{i=1}^{T}\mathbf{R}_i\mathbf{S}_i$. The user then averages the received results to obtain
\begin{align}
    \frac{1}{N} \sum_{i=1}^{N} [[\mathbf{C}]]_i = \sum_{l=1}^{K} \mathbf{A}_l \mathbf{B}_l + \sum_{l=1}^{T}\mathbf{R}_l\mathbf{S}_l. \label{user_source_same}
\end{align}
Since the user has access to the secret keys it has used to encrypt the partitions, it can subtract their product, $\sum_{i=1}^{T}\mathbf{R}_i\mathbf{S}_i$, to obtain the desired result. We note here that the product of the random matrices $\mathbf{R}$ and $\mathbf{S}$ needs to be pre-computed by the user in order to obtain the desired matrix product from Eq. \eqref{user_source_same}. If $T<< N-T$, this would require much less computational resources compared to multiplying the matrices $\mathbf{A}$ and $\mathbf{B}$, or alternatively these multiplications can be done in advance in an offline manner, and stored at the user, and hence, this computation does not affect the computation latency.

\begin{theorem}
A $(N,N-T,T)$ SDMM scheme for two input matrices is achievable for the special case, or in other words, we can securely multiply two matrices using $N$ servers, $T$ of which may collude, with $N > T$, with an optimal upload cost of $\frac{N}{N-T}$.
\end{theorem}

\begin{proof}
It follows from the definition of an SDMM scheme that the upload cost of the $(N,N-T,T)$ SDMM scheme proposed above is $\chi_{UL}=\frac{N}{N-T}$. Moreover, it is proved in \cite{Kakar2019UplinkDownlinkTI} that the optimal upload cost of an SDMM scheme is lower bounded by $\frac{N}{N-T}$. This proves the optimality of the proposed scheme in terms of the upload cost.
\end{proof}

\subsection{Proof of security against colluding servers}
We next prove that the proposed scheme is secure, i.e., the security constraint \eqref{security_constraint} is satisfied. We point out that there is no exchange of messages between the servers, that is, $\mathcal{M}_{i,i'}=\phi$ for all $i,i' \in [N], i\neq i'$. For any $\mathcal{L}$ with $\vert \mathcal{L} \vert = T$, we have
\begin{align}
   \hspace{-0.4cm} I(\mathbf{A},\mathbf{B}; [[\mathbf{A}]]_{\mathcal{L}}, [[\mathbf{B}]]_{\mathcal{L}})
    &= H([[\mathbf{A}]]_{\mathcal{L}}, [[\mathbf{B}]]_{\mathcal{L}}) - H([[\mathbf{A}]]_{\mathcal{L}},[[\mathbf{B}]]_{\mathcal{L}} \vert \mathbf{A}, \mathbf{B}) \\
    &  \overset{a}{=} H([[\mathbf{A}]]_{\mathcal{L}},[[\mathbf{B}]]_{\mathcal{L}}) - H(\mathbf{R}, \mathbf{S})\\
    & \overset{b}{=} H([[\mathbf{A}]]_{\mathcal{L}}) + H([[\mathbf{B}]]_{\mathcal{L}}) - H(\mathbf{R}) - H(\mathbf{S})\\
    & \leq \sum_{i \in \mathcal{L}} H([[\mathbf{A}]]_{i}) + \sum_{i \in \mathcal{L}} H([[\mathbf{B}]]_{i}) - \frac{mnT}{K} \log \vert \mathbb{F}_q \vert - \frac{pnT}{K} \log \vert \mathbb{F}_q \vert\\
    & =  0, \label{proof_final_step}
\end{align}
where $(a)$ follows from \eqref{polyA} and \eqref{polyB}; $(b)$ follows from the fact that $[[\mathbf{A}]]_{\mathcal{L}}$, $[[\mathbf{B}]]_{\mathcal{L}}$, $\mathbf{R}$ and  $\mathbf{S}$ are independent of each other; \eqref{proof_final_step} follows because the elements of the secret shares $[[\mathbf{A}]]_i \in \mathbb{F}_q^{m \times \frac{n}{K}}$ and $[[\mathbf{B}]]_i \in \mathbb{F}_q^{\frac{n}{K}\times p}, \forall i \in \mathcal{L}$ are independent and uniformly distributed in $\mathbb{F}_q$. Hence, $\sum_{i \in \mathcal{L}} H([[\mathbf{A}]]_{i}) = \frac{mnT}{K}\log \vert \mathbb{F}_q \vert$ and $\sum_{i \in \mathcal{L}} H([[\mathbf{B}]]_{i}) = \frac{pnT}{K}\log \vert \mathbb{F}_q \vert$, where $\vert \mathbb{F}_q \vert$ denotes the cardinality of the field $\mathbb{F}_q$.

\subsection{Security against the user}\label{security_user}
When the input matrices are generated by distributed source nodes, the user must not gain additional information beyond the result of the computation. Most existing schemes \cite{2018arXiv180600469C, 8849446, DBLP:journals/corr/abs-1806-00939, Aliasgari:ISIT:19} which rely on polynomial interpolation, do not satisfy this condition, since they are designed for the case when the data is generated by the user. If those schemes are employed, once $\mathbf{C}(x)=\mathbf{A}(x)\mathbf{B}(x)$ is interpolated by the user after receiving the evaluations of $\mathbf{C}(x)$ at a number of points equal to the number of terms in $\mathbf{C}(x)$, it can be factorized to obtain information about $\mathbf{A}(x)$ and $\mathbf{B}(x)$, thus leaking additional information to the user. Generally, factorizations of matrices are not unique in any field, but even then they can narrow the search space significantly in certain cases, thus leaking partial information, and in some cases, factorization leaks complete information of $\mathbf{A}$ and $\mathbf{B}$ by providing unique factors. For example, consider that $m=p=1$ and $K=n$, and the polynomials $\mathbf{A}(x)$ and $\mathbf{B}(x)$ are irreducible polynomials over $\mathbb{F}_q$. Then, on recovering the polynomial $\mathbf{C}(x) = \mathbf{A}(x)\mathbf{B}(x)$, there is a unique factorization of $\mathbf{C}(x)$ in $\mathbb{F}_q$ which provides the factors as the two constituent irreducible polynomials. Our scheme is robust to such information leakages.

We point out that besides the constant term, the user can recover the sum of the coefficients of $x^{N-i}$ and $x^{-i}, i \in [N-1]$, due to aliasing in Eq. \eqref{poly_C}, because $\alpha_N^{-i} = \alpha_N^{N-i}$. The polynomial after considering the aliasing effect is shown in Eq. \eqref{aliased_polyC}. These residual terms may or may not leak partial information about the input matrices to the user. To show this, first we mention a result from \cite{2019arXiv190806957J} that determines the entropy of the product of two matrices depending on their dimensions:

\begin{lemma}\label{entropy_product}
Let $\mathbf{P}, \mathbf{Q}$ be random matrices independently and uniformly distributed over $\mathbb{F}^{m\times n}_q$ and $\mathbb{F}_q^{n\times p}$, respectively. As $q \rightarrow \infty$, we have
\begin{align}
    H(\mathbf{PQ})&= \left\{ \begin{array}{cc}
        mp & n \geq \min(m,p) \\
        mn + np - n^2 & n < \min(m,p)
    \end{array} \right. ,\\
    H(\mathbf{PQ} \mid \mathbf{P}) &= \min\{ mp, np \}\\
     H(\mathbf{PQ} \mid \mathbf{Q}) &= \min\{ mp, mn \}
\end{align}
in $q-$ary units.
\end{lemma}

A relaxed notion of security against the user can be shown to be satisfied under certain conditions. The coefficient of each power of $x$, where the coefficient of the $j^{th}$ power of $x$ is denoted by $\mathbf{C}_j$, is the sum of at least $K$ matrix products. Therefore, the coefficient of $x^4$, for example, can be written as:
\begin{align}
   \mathbf{C}_4 = \mathbf{A}_3\mathbf{S}_1 + \mathbf{R}_1\mathbf{S}_2 + \mathbf{R}_2 \mathbf{B}_1  = \left[ \begin{array}{ccc}
        \mathbf{A}_3  & \mathbf{R}_1 & \mathbf{R}_2 
     \end{array} \right] \left[ \begin{array}{c}
          \mathbf{S}_1 \\
           \mathbf{S}_2 \\
           \mathbf{B}_1 
     \end{array} \right],
\end{align}
which is a product of two matrices of dimensions $m \times n$ and $n \times p$, respectively. If $ n \geq \max (m,p)$, and matrices $\mathbf{A}$ and $\mathbf{B}$ are i.i.d. uniform, then from Lemma \ref{entropy_product} we have $H(\mathbf{A}_3\mathbf{S}_1 + \mathbf{R}_1\mathbf{S}_2 + \mathbf{R}_2 \mathbf{B}_1) = H(\mathbf{A}_3\mathbf{S}_1 + \mathbf{R}_1\mathbf{S}_2 + \mathbf{R}_2 \mathbf{B}_1 \mid \mathbf{A}, \mathbf{R}) = mp$ in q-ary units. Therefore, we have $I(\mathbf{A}, \mathbf{R}; \mathbf{A}_3\mathbf{S}_1 + \mathbf{R}_1\mathbf{S}_2 + \mathbf{R}_2 \mathbf{B}_1) = 0$. Similarly, we have $I(\mathbf{B}, \mathbf{S}; \mathbf{A}_3\mathbf{S}_1 + \mathbf{R}_1\mathbf{S}_2 + \mathbf{R}_2 \mathbf{B}_1) = 0$. Thus the coefficient of $x^4$, that is $\mathbf{C}_4$, is i.i.d. uniformly distributed and independent of the input matrices. Similarly, we have $I(\mathbf{A}, \mathbf{R}; \mathbf{C}_j) = 0$ and $I(\mathbf{B}, \mathbf{S}; \mathbf{C}_j) = 0$ for all $j = 1,\ldots, N-1$. On the other hand, if $p \leq n \leq m$, we have $I(\mathbf{A}, \mathbf{R}; \mathbf{A}_3\mathbf{S}_1 + \mathbf{R}_1\mathbf{S}_2 + \mathbf{R}_2 \mathbf{B}_1) = mp - np > 0$, and $I(\mathbf{B}, \mathbf{S}; \mathbf{A}_3\mathbf{S}_1 + \mathbf{R}_1\mathbf{S}_2 + \mathbf{R}_2 \mathbf{B}_1) = 0$, while if $m \leq n \leq p$, we have $I(\mathbf{A}, \mathbf{R}; \mathbf{A}_3\mathbf{S}_1 + \mathbf{R}_1\mathbf{S}_2 + \mathbf{R}_2 \mathbf{B}_1) = 0$, and $I(\mathbf{B}, \mathbf{S}; \mathbf{A}_3\mathbf{S}_1 + \mathbf{R}_1\mathbf{S}_2 + \mathbf{R}_2 \mathbf{B}_1) = mp - mn > 0$.

When $n < \min (m,p)$, or in general, when matrices $\mathbf{A}$ and $\mathbf{B}$ are not i.i.d. uniform randomly distributed, the user can infer partial information about both $\mathbf{A}$ and $\mathbf{B}$, and instead the following procedure can be implemented where the servers exchange shares of their computed results with each other in a secure way to discard the residual terms, thus ensuring security against the user.
The servers exchange shares of their results from the computation phase in a secure way, similarly to \cite{8437651}, so that each server ends up with a $(N,N-T,T)$ left-share of matrix $\mathbf{C}$, which can then be delivered to the user. To do this, 
\begin{itemize}
    \item Server $i$ generates $(N,N-T,T)$ left-shares of $[[\mathbf{C}]]_i$, evaluated on the $N^{th}$ roots of unity $\alpha_N^{j-1}, \forall j\in [N]$, with $\alpha_N$ being a primitive $N^{th}$ root of unity in $\mathbb{F}_q$, and enumerated as $[[\mathbf{C}]]_{i,j}^{L}=\left[\left[[[\mathbf{C}]]_i \right]\right]_j^L, \forall j\in [N]$.
    \item Server $i$ sends the left-share $[[\mathbf{C}]]_{i,j}^{L}$ to server $j$. The privacy requirement against the servers is satisfied, since any $T$ colluding servers cannot gain any information about server $i$'s share $[[\mathbf{C}]]_i$ from the left-shares received in the communication phase.
    \item Server $j$ averages the received left-shares $[[\mathbf{C}]]_{i,j}^{L}, \forall i\in [N]$, to obtain the $(N,N-T,T)$ left-share $[[\mathbf{C}]]^L_{j}$. 
\end{itemize}

To see the correctness of the above procedure, note that, for given $N,K$ and $T$ values, the secret sharing scheme is linear for both left and right shares; that is, $[[\mathbf{A}]] + [[\mathbf{B}]] = [[\mathbf{A}+\mathbf{B}]]$. Therefore, following from Eq. \eqref{constant_term}, we have
\begin{align}
\frac{1}{N}\sum_{i=1}^{N} [[ \mathbf{C} ]]_{i,j}^{L} &= \frac{1}{N} \sum_{i=1}^{N}\left[\left[[[\mathbf{C}]]_{i}\right]\right]_{j}^L\\
&=\left[\left[\mathbf{C}\right]\right]_j^L.
\end{align}

\subsection{Other performance metrics}
\subsubsection{Encoding complexity}\label{encoding complexity}
To compute $N$ evaluations of the matrix polynomial $\mathbf{A}(x)$ on the roots of unity, the source nodes perform $\frac{mn}{K}$ N-point FFTs, which involve $O(\frac{mnN\log N}{K})$ finite field operations in $\mathbb{F}_q$. For $T=0$, the complexity is $O(mn\log N)$, that is, it is a logarithmic rate of growth with respect to $N$. The analysis is similar for computing the evaluations of matrix polynomial $\mathbf{B}(x)$.

\subsubsection{Download cost}
Download cost is the normalized number of bits that need to be downloaded by the user from the servers to reconstruct the computation result. It is defined as
\begin{align}
    \chi_{DL}=\frac{\sum_{i=1}^{N} H([[\mathbf{C}]]_i)}{H( \mathbf{\mathbf{C}})}.
\end{align}

As shown in \cite{2019arXiv190806957J}, the download cost depends on the dimensions of the data matrices. 
We can compute the download cost of our scheme for three different cases.
\begin{itemize}
    \item If $\min(m,p) \leq  n \leq K\min(m,p)$, the download cost is
\begin{align}
    \chi_{DL}&=\frac{N(\frac{mn}{K}+\frac{np}{K}-\frac{n^2}{K^2})}{mp}\\
    &=\frac{Nn(m+p-\frac{n}{K})}{Kmp}.
\end{align}
Assuming $m=p$, we have $\chi_{DL}\approx\frac{2nN}{mK} < \frac{2N}{K}$.
\item If $n \geq K\min(m,p)$, the download cost is
\begin{align}
     \chi_{DL}&=\frac{Nmp}{mp}\\
    &=N.
\end{align}
\item If $n<\min(m,p)$, the download cost is 
\begin{align}
    \chi_{DL}&=\frac{N(\frac{mn}{K}+\frac{np}{K}-\frac{n^2}{K^2})}{mn+np-n^2}\\
    &= \frac{N(m+p-\frac{n}{K})}{K(m+p-n)}.
\end{align}
For $m,p \rightarrow \infty$, $\chi_{DL} \rightarrow \frac{N}{K}= \frac{N}{N-T}$. This is the optimal download cost.
\end{itemize}

\subsubsection{Decoding complexity} Since the decoding requires computing the sum of the received results, the decoding complexity of the proposed scheme is negligible.
This is an important advantage of the proposed scheme compared to existing polynomial coding schemes in the literature, which require polynomial interpolation.

\begin{table}[htbp]
    \centering
    \begin{tabular}{c|c|c|c|c} 
       \textbf{Scheme}  & \textbf{Upload cost} & \textbf{Download cost} & \textbf{Encoding complexity} & \textbf{Decoding complexity}  \\
       \hline
       Secure MatDot \cite{Aliasgari:ISIT:19} & $\frac{2N}{N-2T+1}$ & $ N = 2(K+T) - 1 $ & $O(\frac{mn}{K} N \log^2 N \log \log N)$ & $O(\frac{mp}{K} N\log^2 N \log \log N)$\\
       \hline
       GASP \cite{8849446}, Secure PolyDot \cite{Aliasgari:ISIT:19}, & & & \\ USCSA \cite{Kakar2019UplinkDownlinkTI}, Nodehi et al. \cite{8437651} & $O(\sqrt{N})$ & $ \frac{N}{N - O(T)}$ & $O(\frac{mn}{\sqrt{N}}N\log^{2}N \log\log N)$ & $O(\frac{mp}{\sqrt{N}} N \log^2 N \log \log N)$\\
       \hline 
        &  & $N = K+2T$, if $n\geq K \min(m,p)$ & \\
       Proposed scheme & $\frac{N}{N-2T}$ & $\frac{Nn(m+p-n/K)}{Kmp}$, if $n \geq \min(m,p)$ & $O(\frac{mnN}{N-2T}\log N)$ & $O(1)$ \\
       & & $\frac{N(m+p-n/K)}{K(m + p - n)}$, if $n \leq \min(m,p)$ &  &  \\
       \hline 
        Proposed scheme using 1 inter-& & & \\server comm. round (Sec. \ref{security_user}) & $\frac{N}{N-2T}$ & $\frac{N}{N-T}$ & $O(\frac{mnN}{N-2T}\log N)$ & $O(\frac{mp}{N-T}N\log N)$
    \end{tabular}
    \caption{Comparative summary of different schemes.}
        \label{table:comparative}
\end{table}

\subsection{Comparison with other schemes}
There exists a trade-off between the upload cost and the download cost depending on the kind of partitioning employed for matrix multiplication. In \cite{Kakar2018RateEfficiencyAS,Kakar2019UplinkDownlinkTI}, the trade-off within the class of schemes that employ row-by-column partitioning is considered. The other class of schemes that employ column-by-row partitioning, like MatDot and our scheme, provide different points on the upload cost-download cost trade-off. A comparative summary is described in Table \ref{table:comparative}. The GASP \cite{8849446}, secure PolyDot \cite{Aliasgari:ISIT:19}, and USCSA \cite{Kakar2019UplinkDownlinkTI} schemes, which employ row-by-column partitioning generally result in a higher upload cost and lower download cost than those employing column-by-row partitioning. Consider $K$ row-wise partitions of $\mathbf{A}$, and $K$ column-wise partitions of $\mathbf{B}$. For $T=0$, the matrices are encoded by evaluating their corresponding polynomials at $N=K^2$ distinct points in $\mathbb{F}_q$. SDMM schemes employing the row-by-column partitioning require $N=O(K^2)$ servers. Hence, the upload cost is $\chi_{UL}=\frac{N}{K}=O(\frac{N}{\sqrt{N}})=O(\sqrt{N})$. Since the complexity of evaluating a polynomial at $N$ points is $O(N\log^2 N\log\log N)$ \cite{BORODIN1974366}, the total encoding complexity is $O(\frac{mn}{\sqrt{N}}N \log^2 N\log\log N)\approx O(mn\sqrt{N}\log^{2}N \log\log N)$. This is significantly larger than the encoding complexity of our scheme, which grows only logarithmically with $N$.

The secure MatDot scheme of \cite{Aliasgari:ISIT:19}, which employs column-by-row partitioning, has an upload cost of $\chi_{UL}=\frac{2N}{N-2T+1}$, which is also larger than that of our scheme by a factor of two (see Fig. \ref{fig:plot}).

\begin{figure}[htbp]
    \centering
    \includegraphics[scale=0.8]{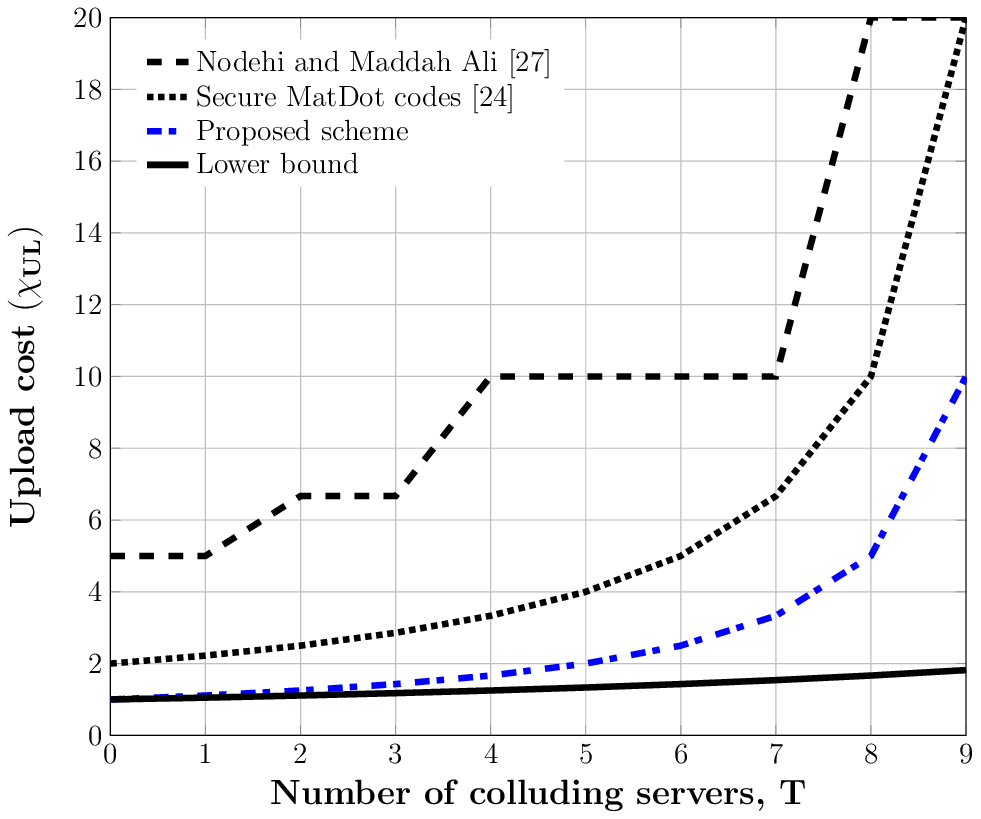}
    \caption{Comparison of the upload cost of SDMM of two matrices for different number of colluding servers out of a total of $N=20$ servers. Only the values of $T\leq 9$ are considered since the proposed scheme is defined only for $T< N/2$. The maximum value of $K$ is calculated for each $T$.
    }
    \label{fig:plot}
\end{figure}

Compared to \cite{8437651}, which also considers the general case with the source nodes separate from the user, our scheme provides significant improvements in terms of the upload cost. The upload cost achieved by the scheme in \cite{8437651} is given by $\chi_{UL}=\frac{\min(2K^2+2T-3,K^2 + KT+T-2)}{K}$, which is significantly higher than that provided by our scheme (see Fig. \ref{fig:plot}). We also highlight that, as opposed to the scheme in \cite{8437651}, our scheme does not require inter-server communication in the computation and communication phases for multiplying two matrices, which significantly reduces both the latency and the complexity.


\section{Straggler mitigation} \label{straggler_mitigation}
The FFT-based scheme in the previous sections does not provide robustness against straggling servers. We need all the evaluations from all the servers to recover the desired result. In this section we present an extension of our FFT-based scheme, which provides a certain level of robustness to straggling servers by incorporating both column-wise and row-wise partitioning of the matrices.

When providing straggler mitigation, typically the goal is to minimize the $\textit{recovery threshold}$, which refers to the minimum integer $r$ such that the computation result can be recovered from any $r$ responsive servers. Hence, instead of specifying the number of servers, and minimizing the upload cost, in this section, we consider arbitrary partitioning of the matrices, and identify the corresponding upload cost and the recovery threshold.

In particular, we employ row-wise partitioning of $\mathbf{A}$ (and column-wise for $\mathbf{B})$ to introduce straggler-robustness. That is, $\mathbf{A}$ is partitioned into a $K_2 \times K_1$ array of equal-sized blocks, and $\mathbf{B}$ is partitioned into a $K_1 \times K_3$ array of equal-sized blocks. The idea is that the secret shares of each row of $\mathbf{A}$ (and each column of $\mathbf{B}$) are generated, while polynomially-coded shares of each column of $\mathbf{A}$ (and each row of $\mathbf{B}$) are generated for straggler robustness. The proposed scheme provides a recovery threshold of $ N - \left( \lceil \frac{N}{K_1+2T} \rceil - K_2K_3\right)$ in general, and $ N - \left( \lceil \frac{N}{K_1+T} \rceil - K_2K_3\right)$ for computing with own data. 

The matrices $\mathbf{A}$ and $\mathbf{B}$ are partitioned into $K_1K_2$ and $K_2K_3$ blocks, denoted by $\mathbf{A}_{i,j}, i \in [K_2], j\in [K_1]$ and $\mathbf{B}_{j,k}, j \in [K_1], k\in [K_3]$, respectively. Matrices $\mathbf{R}_{i,j}, i \in [K_2], j \in [T]$, and $\mathbf{S}_{j,k}, j\in [T], k\in [K_3]$, are generated, whose entries are i.i.d. uniform random variables from $\mathbb{F}_q$. The matrices $\mathbf{R}$ and $\mathbf{S}$, defined as
\begin{align}
     \mathbf{R}&=\left[ \begin{array}{ccc}
         \mathbf{R}_{1,1}  & \cdots & \mathbf{R}_{1,T} \\
         \vdots  & \ddots & \vdots \\ 
           \mathbf{R}_{K_2,1}  & \cdots & \mathbf{R}_{K_2,T}\\
    \end{array} \right] ,\label{R_general}
    \end{align}
    \begin{align}
     \mathbf{S} &= \left[
     \begin{array}{ccc}
          \mathbf{S}_{1,1}   & \cdots & \mathbf{S}_{1,K_3} \\
         \vdots & \ddots  & \vdots \\
         \mathbf{S}_{T,1} & \cdots & \mathbf{S}_{T,K_3} \\
     \end{array}
     \right],\label{S_general}
\end{align}
are used as secret keys to encode the matrices $\mathbf{A}$ and $\mathbf{B}$, respectively, by first appending them in the following manner:
\begin{align}
    \Tilde{\mathbf{A}}&=\left[ \begin{array}{cc}
         \mathbf{A} & \mathbf{R}
    \end{array} \right] ,\label{A_general}\\
     \Tilde{\mathbf{B}} &= \left[
     \begin{array}{c}
         \mathbf{B}    \\
         \mathbf{S}
         \end{array}
     \right], \label{B_general}
\end{align}
and then constructing the following multivariate polynomials:
\begin{align}
   \mathbf{A}(x_1,x_2)&= \sum_{i=1}^{K_2} \sum_{j=1}^{K_1} \mathbf{A}_{i,j} x_2^{i-1} x_1^{j-1} + \sum_{i=1}^{K_2} \sum_{j=1}^{T} \mathbf{R}_{i,j} x_2^{i-1} x_1^{K_1+j-1} \label{A_multipartition}\\
  \mathbf{B}(x_1,x_2)&= \sum_{j=1}^{K_1} \sum_{k=1}^{K_3} \mathbf{B}_{j,k} x_1^{-(j-1)} x_2^{(k-1)K_2}  + \sum_{j=1}^{T} \sum_{k=1}^{K_3} \mathbf{S}_{j,k} x_1^{-(K_1+T+j-1)} x_2^{(k-1)K_2}.
\end{align}

The product polynomial $\mathbf{C}(x_1,x_2)$ is given by
\begin{align}
   \mathbf{C}(x_1,x_2)&= \sum_{i=1}^{K_2} \sum_{k=1}^{K_3} \left(\sum_{j=1}^{K_1}\mathbf{A}_{i,j}\mathbf{B}_{j,k} \right) x_2^{K_2(k-1)+i-1} + \text{(terms with non-zero powers of $x_1$)}.
\end{align}

\subsection{Sharing phase}
Let there be a total of $N\geq N_1N_2$ servers, where $N_1=K_1+2T$ and $N_2 \geq K_2K_3$. The user evaluates polynomials $\mathbf{A}(x_1,x_2)$ and $\mathbf{B}(x_1,x_2)$ at points $(x_1,x_2)=(\alpha_{N_1}^r, \beta_s), r=0,\ldots, N_1-1, s=0,\ldots, N_2 -1$, where $\alpha_{N_1}$ is a primitive $N_1^{th}$ root of unity in $\mathbb{F}_q$, and $\beta_s, \forall s\in [N_2]$ are distinct points in $\mathbb{F}_q$.

The user sends $\mathbf{A}(\alpha_{N_1}^r, \beta_s)$ and $\mathbf{B}(\alpha_{N_1}^r, \beta_s)$ to server $i$, where $i=(s-1)N_1+r-1,\forall  r\in [N_1],\forall s \in [N_2]$. 

\subsection{Computation phase}
 Server $i$, where $i=(s-1)N_1+r-1$, computes $\mathbf{C}(\alpha_{N_1}^r, \beta_s)=\mathbf{A}(\alpha_{N_1}^r, \beta_s)\mathbf{B}(\alpha_{N_1}^r, \beta_s)$, and sends this share to the user as soon as it is computed.
 
\subsection{Reconstruction phase}
The user collects and computes the averages of all the evaluations of $\mathbf{C}(x_1,x_2)$ on the points having the same $x_2$ coordinate, which removes the terms with non-zero exponents of $x_1$. Therefore, for every $x_2 \in [N_2]$, the user obtains
\begin{align}
    \hspace{-0.9cm}\mathbf{f}(x_2)&=\frac{1}{N_1} \sum_{i=0}^{N_1-1} \mathbf{C}(\alpha_{N_1}^i,x_2)\\ 
     &= \sum_{i=1}^{K_2} \sum_{k=1}^{K_3} \left(\sum_{j=1}^{K_1}\mathbf{A}_{i,j}\mathbf{B}_{j,k} \right)  x_2^{K_2(k-1)+i-1}.
\end{align}

The user then interpolates the polynomial $\mathbf{f}(x_2)$ from any $K_2K_3$ evaluations to obtain $\sum_{j=1}^{K_1}\mathbf{A}_{i,j}\mathbf{B}_{j,k} $ for all $(i,k)\in [K_2]\times [K_3]$, therefore obtaining the final result $\mathbf{C}=\mathbf{AB}$.

\subsection{Special case: computation with own data}
For the special case when the data matrices belong to the user, the encoding polynomials are constructed by the user as follows.
\begin{align}
     \mathbf{A}(x_1,x_2)&= \sum_{i=1}^{K_2} \sum_{j=1}^{K_1} \mathbf{A}_{i,j} x_2^{i-1} x_1^{j-1}  + \sum_{i=1}^{K_2} \sum_{j=1}^{T} \mathbf{R}_{i,j} x_2^{i-1} x_1^{K_1+j-1}\\
   \mathbf{B}(x_1,x_2)&= \sum_{j=1}^{K_1} \sum_{k=1}^{K_3} \mathbf{B}_{j,k} x_1^{-(j-1)} x_2^{K_2(k-1)}  + \sum_{j=1}^{T} \sum_{k=1}^{K_3} \mathbf{S}_{j,k} x_1^{-(K_1+j-1)} x_2^{K_2(k-1)}.
\end{align}

By this construction, the matrices $\mathbf{R}$ and $\mathbf{S}$ also appear in the constant coefficient of the polynomials $\mathbf{f}(x_2)$, and the user must pre-compute the products of the private random matrices in order to obtain the desired result.

\begin{remark} \label{group_rt_remark}

Note that, to reconstruct the desired result, the user requires all $N_1$ evaluations on the points $(x_1,x_2)$ for $x_1=1,\alpha_{N_1},\ldots, \alpha_{N_1}^{N_1-1}$ and constant $x_2$. In other words, for any $K_2K_3$ values of $s$, the results from all the groups of servers $[(s-1)N_1 : sN_1 -1]$ are necessary and sufficient for reconstructing the desired result. Therefore, this scheme provides a slightly weaker notion of `group-wise' recovery threshold, that is, the minimum integer $r$ such that any $r$ groups of servers as defined above are sufficient for recovering the computation result. This represents the best-case recovery threshold. The proposed scheme provides a group-recovery threshold of $K_2K_3(K_1+2T)$ in general, and $K_2K_3(K_1+T)$ for computation with own data. 
\end{remark}

\begin{remark}
If at least $ r = N - \left( \lceil \frac{N}{K_1+2T} \rceil - K_2K_3\right)$ number of servers are responsive, that is, there is at most one unresponsive server in each of $\left( \lceil \frac{N}{K_1+2T} \rceil - K_2K_3\right)$ groups of $K_1+2T$ servers, while all the remaining servers are responsive, then it is guaranteed that $K_2K_3$ groups of $K_1 + 2T$ servers, as defined in Remark \ref{group_rt_remark}, are responsive. Therefore, the recovery threshold is given by $r= N - \left( \lceil \frac{N}{K_1+2T} \rceil - K_2K_3 \right)$. However, this is the worst-case recovery threshold, and it is often likely that the required number of groups are responsive with a smaller total number of responsive servers. The average recovery threshold would lie between the group-recovery threshold and the recovery threshold. 
\end{remark}

\subsection{Performance Analysis:}
\begin{itemize}
    \item \textbf{Upload Cost:} We have $\chi_{UL}=\frac{N(K_3mn+K_2np)}{K_1K_2K_3(mn+np)}$.
    \item \textbf{Encoding complexity:} The complexity of encoding matrix $\mathbf{A}$ for $N$ servers is $O\left(\frac{mn}{K_1K_2}\left(N_1\log N_1\right)\left(N_2\log^2 N_2 \log\log N_2\right) \right)$. The encoding complexity for matrix $\mathbf{B}$ can be computed similarly.
    \item \textbf{Decoding complexity:} The decoding requires addition, which has negligible complexity compared to multiplication in a finite field, followed by interpolation of a polynomial with $K_2K_3$ terms, which has complexity $O(n\log^2 n \log\log n)$ for $n=K_2K_3$ \cite{BORODIN1974366}. 
    Note that the decoding complexity of secure PolyDot scheme in \cite{Aliasgari:ISIT:19} is given by the complexity of interpolating a polynomial with $K_2K_3(K_1+T)+K_2K_1 + K_2T-1$ terms, while our scheme requires interpolating a polynomial with $K_2K_3$ terms, which is an order of magnitude smaller.
    \item \textbf{Field size:} The field must have the $N_1^{th}$ roots of unity, and must be big enough to have $K_2K_3$ distinct elements. Therefore, $N_1\vert q-1$ and $q>K_2K_3$ must be satisfied. 
\end{itemize}


\section{Secure distributed multiplication of multiple matrices}\label{multiple_matrices}

In this section, we consider the multiplication of multiple matrices. Such computations, known as matrix chain multiplications, occur in many applications in signal processing, graph theory, and network analysis. We extend our scheme proposed in Section \ref{polymat_eg} to implement multiplication of multiple matrices, that is, given $\Gamma$ source nodes generating the matrices $\mathbf{A}^{(1)},\ldots, \mathbf{A}^{(\Gamma)}$, the user wants to obtain the product $\mathbf{C}=G(\mathbf{A}^{(1)},\ldots, \mathbf{A}^{(\Gamma)})=\mathbf{A}^{(1)}\cdots \mathbf{A}^{(\Gamma)}$ securely from distributed computation over $N$ available servers. A naive method would be for the servers to securely compute the multiplication of two matrices at a time using the proposed scheme in Section \ref{polymat_eg}, and send the results of the computations to the user so that he reconstructs the intermediate computation results, that is, the product of a subset of the matrices, before re-encoding the intermediate result and sending its secret shares to the servers to multiply with the next matrix. However, such a naive method incurs an unnecessary amount of communication cost between the user and the servers, and also requires the user to re-encode the intermediate matrices multiple times, thus increasing the latency of the computation. The naive method also leaks information about the intermediate computations to the user. Next, we propose a more efficient alternative, that also satisfies the user privacy constraint.

The following scheme proceeds iteratively in multiple rounds by obtaining the shares of $\mathbf{C}^{(\gamma)} \triangleq \mathbf{A}^{(1)}\cdots \mathbf{A}^{(\gamma)}$, denoted by $[[\mathbf{C}^{(\gamma)}]]_i \triangleq [[\mathbf{A}^{(1)}\cdots \mathbf{A}^{(\gamma)}]]_i, \forall i \in [N]$, in the $(\gamma-1)^{th}$ round, for $\gamma \in [\Gamma]$. 

\subsection{Sharing phase}
For $K=N-2T$, the $(N,N-2T,T)$ left-shares $[[\mathbf{A}^{(1)}]]_i^L$ and right-shares $[[\mathbf{A}^{(\gamma)}]]_i^{R}, \ \forall \gamma \in [2:\Gamma]$, are sent to server $i$.


\subsection{Computation phase} In the computation phase of the $\gamma^{th}$ round, where $\gamma \in [\Gamma-1]$, server $i\in [N]$ computes the shares $[[\mathbf{H}^{(\gamma + 1)}]]_i = [[\mathbf{C}^{(\gamma)}]]_i^L [[\mathbf{A}^{(\gamma + 1)}]]_i^R, \ \forall \gamma \in [\Gamma-1]$. The secret shares $[[\mathbf{H}^{(\gamma + 1)}]]_i$ are evaluations on the $N^{th}$ roots of unity of a polynomial $\mathbf{H}^{(\gamma+1)}(x)$, which is similar to the product polynomial obtained in Eq. \eqref{poly_C}, and whose constant term is the matrix $\mathbf{C}^{(\gamma+1)}=\mathbf{C}^{(\gamma)}\mathbf{A}^{(\gamma+1)}$, while the remaining terms are uniformly distributed random matrices. 

\subsection{Communication phase} In the communication phase, the servers exchange shares of their results from the computation phase in a secure way, similarly to \cite{8437651}, to convert their secret shares $[[\mathbf{H}^{(\gamma+1)}]]$ to $(N,N-2T,T)$ left-shares of matrix $\mathbf{C}^{(\gamma + 1)}$. To do this, 
\begin{itemize}
    \item Server $i$ generates $(N,N-2T,T)$ left-shares of $[[\mathbf{H}^{(\gamma + 1)}]]_i$, evaluated on the $N^{th}$ roots of unity $\alpha_N^{j-1}, \forall j\in [N]$, with $\alpha_N$ being a primitive $N^{th}$ root of unity in $\mathbb{F}_q$, and enumerated as $[[\mathbf{H}^{(\gamma + 1)}]]_{i,j}^{L}=\left[\left[[[\mathbf{H}^{(\gamma + 1)}]]_i \right]\right]_j^L, \forall j\in [N]$.
    \item Server $i$ sends the left-share $[[\mathbf{H}^{(\gamma + 1)}]]_{i,j}^{L}$ to server $j$. The privacy requirement against the servers is satisfied, since any $T$ colluding servers cannot gain any information about server $i$'s share $[[\mathbf{H}^{(\gamma + 1)}]]_i$ from the left-shares received in the communication phase.
    \item Server $j$ averages the received left-shares $[[\mathbf{H}^{(\gamma + 1)}]]_{i,j}^{L}, \forall i\in [N]$, to obtain the $(N,N-2T,T)$ left-share $[[\mathbf{C}^{(\gamma + 1)}]]^L_{j}$. 
\end{itemize} 

To see the correctness of the above procedure, note that, for given $N,K$ and $T$ values, the secret sharing scheme is linear for both left and right shares; that is, $[[\mathbf{A}]] + [[\mathbf{B}]] = [[\mathbf{A}+\mathbf{B}]]$. Therefore, following from Eq. \eqref{constant_term}, we have
\begin{align}
\frac{1}{N}\sum_{i=1}^{N} [[ \mathbf{H}^{(\gamma+1)} ]]_{i,j}^{L} &= \frac{1}{N} \sum_{i=1}^{N}\left[\left[[[\mathbf{H}^{(\gamma+1)}]]_{i}\right]\right]_{j}^L\\
&=\left[\left[\mathbf{C}^{(\gamma+1)}\right]\right]_j^L.
\end{align}



The scheme proceeds in a recursive manner, looping back to the computation phase for the $(\gamma+1)^{th}$ round.

\subsection{Reconstruction phase} At the end of the $(\Gamma-1)^{th}$ round, the servers have access to the secret shares of the matrix $\mathbf{C}^{(\Gamma)}=\mathbf{A}^{(1)}\cdots \mathbf{A}^{(\Gamma)}$. The servers send the shares $[[\mathbf{C}^{(\Gamma)}]]_i, \forall i \in [N]$, to the user, which then computes the average of these shares to obtain
\begin{align}\label{Gamma_matrices_result}
    \frac{1}{N} \sum_{i=1}^{N} [[\mathbf{C}^{(\Gamma)}]]_i=\mathbf{A}^{(1)}\cdots \mathbf{A}^{(\Gamma)}.
\end{align}

\subsection{Performance analysis}
\textbf{Upload cost:} Each server receives one $(N,N-2T,T)$ share of all $\mathbf{A}^{(\gamma)}, \gamma \in [\Gamma]$. Thus we have
\begin{align}
    \chi_{UL}&=\frac{\sum_{i=1}^{N}\sum_{\gamma=1}^{\Gamma} \frac{ H([[\mathbf{A}^{(\gamma)}]]_i)}{K}}{\sum_{\gamma=1}^{\Gamma}H(\mathbf{A}^{(\gamma)})} \\
    &=\frac{N}{N-2T}.
\end{align}

\textbf{Encoding complexity:} Since the encoding involves the computation of FFT, the complexity is $O(m_{\gamma}m'_{\gamma}\frac{N}{K}\log N)$ for encoding the matrix $\mathbf{A}^{(\gamma)}, \forall \gamma\in [\Gamma]$.

\textbf{Complexity of the communication phase:} The complexity of server $i$ generating $(N,N-2T,T)$ shares of $[[\mathbf{H}^{(\gamma + 1)}]]_i$ is $O(m_{\gamma}m'_{\gamma + 1}\frac{N}{K}\log N)$, while the complexity of server $j$ averaging the received shares from the other servers is ignored as it requires only addition. The cost of inter-server communication is $\frac{N-1}{N-2T}$ per server, since each server communicates $N-1$ shares to the remaining servers.

\textbf{Decoding complexity:} The partial decoding in the communication phase of each intermediate round, as well as the final decoding of the result after the completion of the $(\Gamma-1)^{th}$ round requires only the addition of the received results from all the servers. Therefore, the decoding complexity is negligible.

\textbf{Security against the user:} The scheme constructed in this section preserves security against the user, since the intermediate computation results are not communicated to the user in any form. In contrast, this security constraint is violated by the naive scheme described at the start of this section, because the user obtains the result of each intermediate computation.

\section{Secure matrix algebra}\label{matrix_algebra}
In this section, we describe algorithms for performing matrix operations besides matrix multiplication that are useful in matrix algebra. Some operations, like matrix inversion, can be reduced to matrix multiplication, and can be implemented with the SDMM schemes described in the preceding sections.
\subsection{Matrix addition and multiplication by a scalar}\label{S_addition}
Addition and scalar multiplication follow easily from the linear nature of the secret sharing scheme. We have $[[\mathbf{A}+\mathbf{B}]]=[[\mathbf{A}]]+[[\mathbf{B}]]$, and $[[c\mathbf{A}]]=c[[\mathbf{A}]]$, where $c \in \mathbb{F}_q$. For addition, the shares must both be either left-shares or right-shares.

\subsection{Changing from $(N,K_1,T_1)$ left-shares to  $(N,K_2,T_2)$ right-shares}\label{change_partition}

\begin{figure*}[htbp]
    \centering
    \includegraphics[scale=0.9]{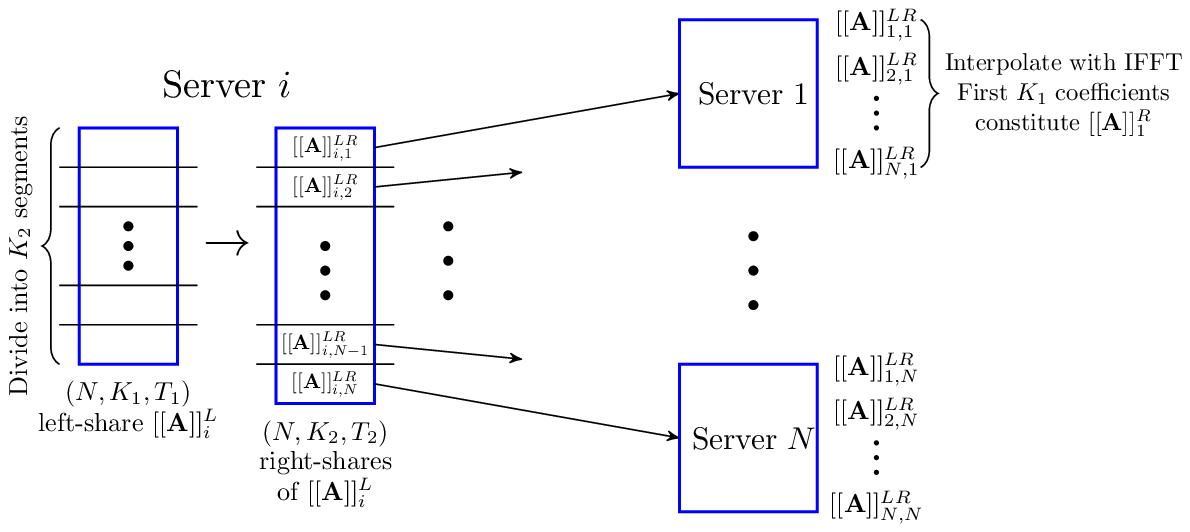}
    \caption{Conversion from left-share to right-share (right-share to left-share conversion can be done similarly).}
    \label{fig:share_conv}
\end{figure*}
Suppose the servers store $(N,K_1,T_1)$ left-shares of matrix $\mathbf{A}$. The goal is to let the servers obtain $(N,K_2,T_2)$ right-shares of matrix $\mathbf{A}$. A special case of this algorithm with $K_1=1$, $K_2=N-2T$ and $T_1=T_2=T$ was used earlier in the communication phase of Section \ref{multiple_matrices}. The procedure takes the following steps (see Fig. \ref{fig:share_conv}):
\begin{enumerate}
    \item Server $i$ generates $(N,K_2,T_2)$ right-shares of $[[\mathbf{A}]]^L_i$ evaluated on the $N^{th}$ roots of unity $\alpha_N^{j-1}, \forall j\in [N]$, and enumerated as $[[\mathbf{A}]]^{LR}_{i,j} = \left[\left[[[\mathbf{A}]]^{L}_{i}\right]\right]_{j}^{R}, \forall j\in [N]$.
    \item Server $i\in [N]$ sends the right-share $[[\mathbf{A}]]^{LR}_{i,j}$ to server $j$.
    \item Server $j$ interpolates the received shares $[[\mathbf{A}]]^{LR}_{i,j}, i= 1,\ldots, N$, using inverse FFT (IFFT) to obtain a polynomial whose first $K_1$ coefficients are $K_1$ column-wise partitions of the share $[[\mathbf{A}]]_{j}^{R}$. Stacking them column-wise gives the share $[[\mathbf{A}]]_{j}^{R}$.
\end{enumerate}

A procedure to convert right-shares to left-shares can be obtained similarly to the above procedure. For example, to convert right-shares to left-shares, step 1 of the above procedure generates $(N,K_2,T_2)$ left-shares of $[[\mathbf{A}]]_i^R$, enumerated as $[[\mathbf{A}]]^{RL}_{i,j} = \left[\left[[[\mathbf{A}]]^{R}_{i}\right]\right]_{j}^{L}, \forall j\in [N]$.

To see the correctness of the above procedure, note that IFFT of a sequence of $(N,K_1,T_1)$ left-shares $\left\{[[\mathbf{A}]]_1^L,\ldots,[[\mathbf{A}]]_N^L\right\}$ gives the sequence $\left\{\mathbf{A}_1,\ldots,\mathbf{A}_{K_1},\mathbf{R}_1,\ldots,\mathbf{R}_{T_1}\right\}$. Similarly, IFFT of a sequence of $(N,K_1,T_1)$ right-shares $\left\{[[\mathbf{A}]]_1^R,\ldots,[[\mathbf{A}]]_N^R\right\}$ gives the sequence $\left\{\mathbf{A}_1,\ldots,\mathbf{A}_{K_1},\mathbf{0}_1,\ldots, \mathbf{0}_{T_1},\mathbf{R}_1,\ldots, \mathbf{R}_{T_1}\right\}$. Since each coefficient of the IFFT sequence is a linear combination of the elements of the input sequence; for all $l\in [K_1]$, the $l^{th}$ coefficient of the IFFT, denoted by $IFFT_l(.)$, of the sequence of shares received by server $j$ is given by
\begin{align}
 IFFT_l \Bigl(& [[\mathbf{A}]]_{1,j}^{LR}, \ldots,  [[\mathbf{A}]]_{N,j}^{LR} \Bigr)\\
   &= IFFT_l \Bigl(\left[\left[ [[\mathbf{A}]]_1^{L} \right]\right]_j^R, \ldots, \left[\left[ [[\mathbf{A}]]_N^{L} \right]\right]_j^R \Bigr)\\
    &=\left[\left[ IFFT_l \left( [[\mathbf{A}]]_1^{L},\ldots, [[\mathbf{A}]]_N^{L} \right) \right]\right]_j^R\\
    & = \left[ \left[ \mathbf{A}_l \right]\right]_j^R,
\end{align}
where $\mathbf{A}_l, l\in [K_1]$ is the $l^{th}$ column-wise partition used for obtaining the original $(N,K_1,T_1)$ left-shares of matrix $\mathbf{A}$. Thus, server $j$ obtains the secret shares $[[\mathbf{A}_l]]_j^R, l =1,\ldots, K_1$.

\begin{remark}
For $K_1\geq 2$, left-shares cannot be directly converted to left-shares using the above procedure. Similarly, right-shares cannot be directly converted to right-shares. However, $(N,1,T_1)$ left-shares can be directly converted to both $(N,K_2,T_2)$ left-shares and right-shares using the above procedure. For converting $(N,1,T_1)$ left-shares to $(N,K_2,T_2)$ left-shares, step $1$ of the above procedure generates $(N,K_2,T_2)$ left-shares of $[[\mathbf{A}]]_i^L$, enumerated as $[[\mathbf{A}]]^{LL}_{i,j} = \left[\left[[[\mathbf{A}]]^{L}_{i}\right]\right]_{j}^{L}, \forall j\in [N]$.
\end{remark}

\subsection{Transpose of a matrix}\label{transpose}
Consider that server $i$ stores a $(N,K_1,T_1)$ left-share of $\mathbf{A}, i\in [N]$. The goal is to obtain a procedure, through which the servers end up with the $(N,K_2,T_2)$ left-shares of $\mathbf{A}^{tr}$ instead. The procedure takes the following steps:
\begin{enumerate}
    \item Server $i \in [N]$ performs the transpose operation on the left-share $[[\mathbf{A}]]^L_i$, to obtain a share $[[\mathbf{A}^{tr}]]_i=([[\mathbf{A}]]_i^L)^{tr}$. We have
    \begin{align}
        [[\mathbf{A}^{tr}]]_i&=([[\mathbf{A}]]_i^L)^{tr}\\
        &=  \sum_{l=1}^{K_1} \mathbf{A}^{tr}_{l}\alpha_N^{(i-1)(l-1)} + \sum_{l=1}^{T_1}\mathbf{R}^{tr}_{l} \alpha_N^{(i-1)(K+l-1)},
    \end{align}
    where $\mathbf{A}^{tr}_l, l=1,\ldots, K_1$ are equivalent to row-wise partitions of $\mathbf{A}^{tr}$. Note that $[[\mathbf{A}^{tr}]]_i$ is neither a left-share nor a right-share, but is equivalent to having a row-wise partitioning of $\mathbf{A}^{tr}$ employed in right-encoding, and having the exponents of the secret keys employed in left-encoding.
    \item Server $i$ then generates $(N,K_2,T_2)$ left-shares of $[[\mathbf{A}^{tr}]]_i$ evaluated on the $N^{th}$ roots of unity $\alpha_N^{j-1}, \forall j \in [N]$, and enumerated as $[[\mathbf{A}^{tr}]]^{L}_{i,j} = \bigl[\bigl[ [[\mathbf{A}^{tr}]]_{i}\bigr]\bigr]_{j}^{L}, \forall j\in [N]$.
    \item Server $i$ sends the left-share $[[\mathbf{A}^{tr}]]_{i,j}^{L}$ to server $j$.
    \item Server $j$ interpolates the received shares $[[\mathbf{A}^{tr}]]^{L}_{i,j}, i= 1,\ldots, N$, using IFFT to obtain the first $K_1$ coefficients, which are $K_1$ row-wise partitions of the share $[[\mathbf{A}^{tr}]]_{j}^{L}$. Stacking them row-wise gives the share $[[\mathbf{A}^{tr}]]_{j}^{L}$.
\end{enumerate}

\subsection{Exponentiation}
Suppose the user wants to compute $G(\mathbf{A})=\mathbf{A}^{r}, r\in \mathbb{N}$. $(N,N-2T,T)$ left-shares of matrix $\mathbf{A}$ are sent to the servers. The servers then implement the scheme to generate $(N,N-2T,T)$ right-shares of matrix $\mathbf{A}$ as described in Section \ref{change_partition}. Thus, the servers now have the shares of matrix $\mathbf{A}$ for left and right multiplication. The servers then implement the scheme for multiple matrix multiplication described in Section \ref{multiple_matrices}. If $r=2^n$, the computation can be done in $\log r$ rounds of computation and communication phases by computing $\mathbf{A}^2, \mathbf{A}^4, \mathbf{A}^8, \ldots, \mathbf{A}^r$ in successive rounds. If $r\neq 2^n$, then consider the binary expansion $\{a_{B-1},\ldots, a_0 \}$ of $r$, such that $r=\sum_{i=0}^{B-1} a_i 2^i$, where $B$ is the maximum number of bits required to represent $r$. Then, the desired result is given by
\begin{align}
    \mathbf{A}^r&= \mathbf{A}^{\sum_{i=0}^{B-1} a_i 2^i}\\
    &=\prod_{i=0}^{B-1} \mathbf{A}^{a_i 2^i}. \label{exp_general}
\end{align}
If the Hamming weight of the binary expansion of $r$ is $h$, then the computation requires $h-1$ extra rounds of computation and communication phases to perform the computation of Eq. \eqref{exp_general}, resulting in a total of $\log r + h -1$ rounds.

\textbf{Upload cost:} $(N,N-2T,T)$ left-shares of matrix $\mathbf{A}$ are uploaded by the source, while the $(N,N-2T,T)$ right-shares of matrix $\mathbf{A}$ are generated in-situ by the share-conversion algorithm, therefore not required to be sent by the source. Therefore, we have $\chi_{UL}=\frac{N}{N-2T}$.

\subsection{Solving the linear system $\mathbf{AX}=\mathbf{B}$ with secure Gaussian elimination}\label{inverse}

The linear system $\mathbf{AX}=\mathbf{B}$, where the elements of $\mathbf{A}$ and $\mathbf{B}$ belong to $\mathbb{F}_q$, can be solved by performing Gaussian elimination (GE) on the augmented matrix $(\mathbf{A}\vert \mathbf{B})$. Setting $\mathbf{B}$ equal to the identity matrix, the solution of the linear system also gives the matrix inverse $\mathbf{A}^{-1}$, if it exists. The GE method performs elementary row operations on the augmented matrix of the linear system, and row interchanges, also called pivoting, to transform the linear system into its row-echelon form. A scheme for secure GE is described in \cite{Bouman2018NewPF}, which takes element-wise secret shares of matrices $\mathbf{A}$ and $\mathbf{B}$ as inputs, and outputs the solution of the linear system. Computing element-wise secret shares of a matrix is equivalent to computing $(N,1,T)$ secret shares of matrices $\mathbf{A}$ and $\mathbf{B}$. Therefore, if the servers store $(N,K,T)$ secret shares of matrix $\mathbf{A}$ and $\mathbf{B}$, they must first be converted to $(N,1,T)$ secret shares using the scheme in Section \ref{change_partition}.

\subsection{Secure matrix inversion}
Besides solving the linear system $\mathbf{AX}=\mathbf{I}$ to compute the matrix inverse, as described in the previous section, a different procedure, inspired from that in \cite{10.1007/3-540-44647-8_7}, is described below, through which the servers start from $(N,K,T)$ right-shares of a square matrix $\mathbf{A}\in \mathbb{F}_q^{m\times m}$, and end up with $(N,K,T)$ left-shares of $\mathbf{A}^{-1}$ instead.
\begin{enumerate}
    \item A uniformly distributed random matrix $\mathbf{\Phi}\in \mathbb{F}_q^{m\times m}$, which is used as the secret key, is secretly shared with the servers. We assume that there is no central entity that can generate and share the secret key securely with the servers. The shares of the secret key are generated by the servers in a decentralized manner as follows: For $i\in [N]$, server $i$ generates a random matrix $\mathbf{\Phi}^{(i)}\in \mathbb{F}_q^{m\times m}$, and generates its $(N,K,T)$ left-shares evaluated on the $N^{th}$ roots of unity $\alpha_N^{j-1}, \forall j\in [N]$. Server $i$ sends the left-share $[[\mathbf{\Phi}^{(i)}]]^L_j$ to server $j \in [N]\setminus \{ i \}$. From the left-shares received by server $j$, it computes the left-share $[[\mathbf{\Phi}]]^L_j=[[\mathbf{\Phi}^{(1)} + \cdots + \mathbf{\Phi}^{(N)}]]^L_j=[[\mathbf{\Phi}^{(1)}]]^L_j+\cdots + [[\mathbf{\Phi}^{(N)}]]^L_j$. Thus, each server obtains a left-share of a common secret key $\mathbf{\Phi}=\mathbf{\Phi}^{(1)} + \cdots + \mathbf{\Phi}^{(N)}$. 
    \item Server $j$ securely computes $[[\mathbf{P}]]_j=[[\mathbf{\Phi A}]]_j=[[\mathbf{\Phi}]]^L_j[[\mathbf{A}]]^R_j$.
    \item The servers reconstruct matrix $\mathbf{P}$ from its secret shares by exchanging their secret shares of matrix $\mathbf{P}$ with every other server, and averaging the received shares, similar to Eq. \eqref{constant_term}. Thus, each server obtains the public matrix $\mathbf{P}=\mathbf{\Phi A}$. The servers gain no information of the matrix $\mathbf{A}$ from the matrix $\mathbf{P}$, therefore satisfying the privacy constraint against the servers.
    \item Each server computes the matrix inverse $\mathbf{P}^{-1}=\mathbf{A}^{-1}\mathbf{\Phi}^{-1}$.
    \item For $j=1,\ldots, N$, server $j$ then obtains a left-share of the inverse of matrix $\mathbf{A}$ as follows: $[[\mathbf{A}^{-1}]]_j^L=\mathbf{P}^{-1}[[\mathbf{\Phi}]]_j^L=[[\mathbf{P}^{-1}\mathbf{\Phi}]]_j^L$.
    \item The servers can now perform further computation on the left-shares obtained, or deliver their left-shares to the user, who then performs IFFT on the received results to obtain the matrix $\mathbf{A}^{-1}$.
\end{enumerate} 
\textbf{Upload cost:} $(N,N-2T,T)$ shares of matrix $\mathbf{A}$ are uploaded by the source for Step 2, where SDMM of random matrix $\mathbf{\Phi}$ and input matrix $\mathbf{A}$ is performed, thus incurring an upload cost of $\chi_{UL}=\frac{N}{N-2T}$.

\subsection{Iterative matrix inversion}
The method for secure matrix inversion introduced above includes an intermediate step (step 4) that involves the inversion of a secure full-size matrix at each server. While the procedure satisfies the privacy requirements, it may contradict with the motivation of distributed computation. Iterative matrix inversion algorithms, for example Newton's method \cite{1328724}, do not involve direct matrix inversions, but instead proceed with matrix multiplications, and therefore, are amenable to efficient distributed implementation. Newton's method, however, provides only an approximation of the matrix inverse. Newton's method for inverting matrices is derived from Newton's method for finding the root of a function. The procedure is as follows \cite{1328724}:
\begin{itemize}
    \item Set $f(\mathbf{X})=\mathbf{A}-\mathbf{X}^{-1}$. Note that the root of $f(\mathbf{X})$ is $\mathbf{A}^{-1}$. Apply Newton's method for finding its root, as follows:
    \begin{align}
        \mathbf{X}_{i+1}=\mu_i \mathbf{X}_i \left( 2\mathbf{I}-\mathbf{A}\mathbf{X}_i \right),
    \end{align}
    where $\mu_i=1$ for all $i>0$. 
    \item \textbf{Choosing the initial estimate $\mathbf{X}_0$:} Quadratic convergence is obtained if $\vert \vert \mathbf{A}\mathbf{X}_0 - \mathbf{I} \vert \vert < 1$. This is satisfied if $\mathbf{X}_0 = \mu_o \mathbf{A}^T$ is picked as the initial estimate, with the value of $\mu_0$ as proposed in \cite{1328724}.
\end{itemize}
The matrix addition and multiplication operations can be performed securely in a distributed manner using the algorithms described in this paper.

\subsection{Computation of arbitrary polynomials:}
The algorithms for securely performing matrix addition, transpose, exponentiation, inverse and multiplication that have been described in this paper allow the user to compute arbitrary matrix polynomials on distributed servers. For example, a function of the following form,
\begin{align}
    G(\mathbf{A}_1,\mathbf{A}_2,\mathbf{A}_3)= \mathbf{A}_1^2\mathbf{A}_2+c\mathbf{A}^{-1}_3,
\end{align}
can be computed securely on distributed servers as follows: first, the servers securely compute $\mathbf{A}_1^2\mathbf{A}_2$ using the scheme for multiple matrix multiplication; then, the servers add the shares $[[\mathbf{A}_1^2\mathbf{A}_2]]$ to the shares $[[c\mathbf{A}^{-1}_3]]$, computed using one of the secure matrix inversion methods described, to finally obtain the secret shares of the final result, which the user receives and decodes to obtain the desired result.

\subsection{Secure learning from local datasets} 
Consider that data from $D$ source nodes, each with a different size of dataset, is used for training a fully connected deep neural network. In a fully connected neural network, the input layer of neurons performs the matrix multiplication $\mathbf{WX}$, where $\mathbf{W}$ is the weight matrix associated with the layer of neurons, and $\mathbf{X}=(\mathbf{X}_1, \ldots, \mathbf{X}_D)$, where $\mathbf{X}_i, i\in [D]$ is the dataset belonging to source node $i$. SDMM schemes in \cite{2018arXiv180600469C,8849446,Kakar2019UplinkDownlinkTI,8437651}, based on column-wise partitioning of matrix $\mathbf{X}$, code across different data points, thus requiring the local datasets of the sources to be encoded at a central location, which leads to a privacy concern. In contrast, the SDMM algorithm proposed in this paper, based on row-wise partitioning of matrix $\mathbf{X}$, encodes each dataset independently. We have
\begin{align}
    [[\mathbf{X}]]^{R} = \left( [[\mathbf{X}_1]]^R, \ldots, [[\mathbf{X}_D]]^R \right),\label{datasets}
\end{align}
that is, each source can deliver the right shares of its dataset to the servers independent of other sources. 

\textbf{Linear regression-} Consider the computation of the MMSE estimate in a linear regression problem, where the optimum estimate is given by $\beta = \left( \mathbf{X}^{tr}\mathbf{X}\right)^{-1} \mathbf{X}^{tr}\mathbf{Y}$. The local datasets $(\mathbf{X}_i,\mathbf{Y}_i)$ can be delivered to the servers similarly to Eq. \eqref{datasets}, and then the left-shares and right-shares of $\mathbf{X}^{tr}$ and $\mathbf{Y}$ can be generated by using the algorithms to convert shares and perform secure transpose described in this section.

\subsection{Achieving the optimal upload and download cost for SDMM}
If we assume that the cost of inter-server communication is negligible compared to the upload and download costs, optimal communication costs can be achieved simultaneously for both upload and download while ignoring the inter-server communication costs. The source nodes upload $(N,N-T,T)$ shares of the input matrices to the servers, resulting in an upload cost of $\chi_{UL}=\frac{N}{N-T}$, which is shown to be a lower bound on the upload cost for SDMM in \cite{Kakar2019UplinkDownlinkTI}. Using the algorithm proposed in Section \ref{change_partition}, the $(N,N-T,T)$ shares can then be converted to $(N,N-2T,T)$ shares for implementing the SDMM algorithm proposed in this paper. The secret shares of the computation results are then converted to $(N,N-T,T)$ shares using the algorithm in Sections \ref{security_user} and \ref{change_partition}. These shares are sent to the user, resulting in a download cost of $\chi_{DL}=\frac{N}{N-T}$, which is known to be the optimal download cost for SDMM \cite{2018arXiv180600469C}. This procedure circumvents the trade-off between the upload cost and the download cost for SDMM schemes, studied earlier in \cite{Kakar2019UplinkDownlinkTI}.

The assumption that the cost of inter-server communication is negligible is justified in many practical scenarios involving computing clusters, where the computing servers are connected with high-speed communication links, while the links between the source nodes and the servers, and between the user and the servers may have limited bandwidth. However, when the inter-server communication costs (delay, bandwidth and/or energy) are non-negligible, the extra rounds of communication among the servers, required for various share conversions, become prohibitive.

\section{Conclusion and Discussion}
In this paper we developed a novel polynomial coded computation scheme achieving a near-optimal performance in terms of the upload cost for SDMM across $N$ servers, any $T$ of which may collude. We also proposed a scheme achieving the optimal upload cost for the special case when the user requesting the computation is also the source of the matrices to be computed upon. The scheme involves evaluating the constructed polynomials at the roots of unity in an appropriate finite field, which is equivalent to taking the discrete Fourier transform of the constituent matrices. The encoding and decoding complexity is also lower than all the other schemes in the literature. For a special case of the data matrices having certain asymptotic dimensions, our scheme also achieves the optimal download cost. We also introduced a method for straggler mitigation, which provides group-wise tolerance to straggling servers. Straggler tolerance is achieved at the expense of an increase in the upload cost. We further extended our scheme to implement multiplication of multiple matrices, while keeping the input matrices and all the intermediate computations secure against any $T$ colluding servers, with a minimal upload cost. This presents a substantial improvement in performance in terms of the upload cost for multiplication of multiple matrices over existing schemes in the literature. Moreover, we described procedures for other common matrix operations, some of which can be reduced to a set of matrix multiplications, thus allowing us to compute arbitrary matrix polynomials.

For future work, methods for securely performing other matrix operations, such as matrix decompositions, on distributed servers will be explored. It would also be interesting to develop schemes for multiplication of arbitrary number of matrices with minimum inter-server communication. From a more practical perspective, an interesting problem to look at is the simultaneous scheduling of the computation and communication phases to minimize the overall latency.

\bibliography{refer}

\begin{thebibliography}{10}
\providecommand{\url}[1]{#1}
\csname url@samestyle\endcsname
\providecommand{\newblock}{\relax}
\providecommand{\bibinfo}[2]{#2}
\providecommand{\BIBentrySTDinterwordspacing}{\spaceskip=0pt\relax}
\providecommand{\BIBentryALTinterwordstretchfactor}{4}
\providecommand{\BIBentryALTinterwordspacing}{\spaceskip=\fontdimen2\font plus
\BIBentryALTinterwordstretchfactor\fontdimen3\font minus
  \fontdimen4\font\relax}
\providecommand{\BIBforeignlanguage}[2]{{%
\expandafter\ifx\csname l@#1\endcsname\relax
\typeout{** WARNING: IEEEtran.bst: No hyphenation pattern has been}%
\typeout{** loaded for the language `#1'. Using the pattern for}%
\typeout{** the default language instead.}%
\else
\language=\csname l@#1\endcsname
\fi
#2}}
\providecommand{\BIBdecl}{\relax}
\BIBdecl

\bibitem{5394944}
S.~{Rane}, W.~{Sun}, and A.~{Vetro}, ``Secure function evaluation based on
  secret sharing and homomorphic encryption,'' in \emph{Annual Allerton Conf.
  on Comm., Control, and Computing}, Sep. 2009, pp. 827--834.

\bibitem{10.1007/978-3-319-19962-7_27}
M.~Yasuda, T.~Shimoyama, J.~Kogure, K.~Yokoyama, and T.~Koshiba, ``Secure
  statistical analysis using rlwe-based homomorphic encryption,'' in
  \emph{Information Security and Privacy}, E.~Foo and D.~Stebila, Eds.\hskip
  1em plus 0.5em minus 0.4em\relax Cham: Springer International Publishing,
  2015, pp. 471--487.

\bibitem{EfficientSecureMatrixMultiplicationOverLWEBasedHomomorphicEncryption}
\BIBentryALTinterwordspacing
D.~H. Duong, P.~K. Mishra, and M.~Yasuda, ``Efficient secure matrix
  multiplication over lwe-based homomorphic encryption,'' \emph{Tatra Mountains
  Mathematical Publications}, vol.~67, no.~1, pp. 69 -- 83, 2016. [Online].
  Available:
  \url{https://content.sciendo.com/view/journals/tmmp/67/1/article-p69.xml}
\BIBentrySTDinterwordspacing

\bibitem{10.1007/978-3-642-54568-9_3}
M.~Yasuda, T.~Shimoyama, J.~Kogure, K.~Yokoyama, and T.~Koshiba, ``Practical
  packing method in somewhat homomorphic encryption,'' in \emph{Data Privacy
  Management and Autonomous Spontaneous Security}.\hskip 1em plus 0.5em minus
  0.4em\relax Berlin, Heidelberg: Springer Berlin Heidelberg, 2014, pp. 34--50.

\bibitem{Bouman2018NewPF}
N.~J. Bouman and N.~de~Vreede, ``New protocols for secure linear algebra:
  Pivoting-free elimination and fast block-recursive matrix decomposition,''
  \emph{IACR Cryptology ePrint Archive}, vol. 2018, p. 703, 2018.

\bibitem{10.1007/3-540-44647-8_7}
R.~Cramer and I.~Damg{\aa}rd, ``Secure distributed linear algebra in a constant
  number of rounds,'' in \emph{Advances in Cryptology --- CRYPTO 2001},
  J.~Kilian, Ed.\hskip 1em plus 0.5em minus 0.4em\relax Berlin, Heidelberg:
  Springer Berlin Heidelberg, 2001, pp. 119--136.

\bibitem{8002642}
K.~{Lee}, M.~{Lam}, R.~{Pedarsani}, D.~{Papailiopoulos}, and K.~{Ramchandran},
  ``Speeding up distributed machine learning using codes,'' \emph{IEEE
  Transactions on Information Theory}, vol.~64, no.~3, pp. 1514--1529, March
  2018.

\bibitem{pmlr-v70-tandon17a}
\BIBentryALTinterwordspacing
R.~Tandon, Q.~Lei, A.~G. Dimakis, and N.~Karampatziakis, ``Gradient coding:
  Avoiding stragglers in distributed learning,'' in \emph{Proceedings of the
  34th International Conference on Machine Learning}, ser. Proceedings of
  Machine Learning Research, D.~Precup and Y.~W. Teh, Eds., vol.~70.\hskip 1em
  plus 0.5em minus 0.4em\relax International Convention Centre, Sydney,
  Australia: PMLR, 06--11 Aug 2017, pp. 3368--3376. [Online]. Available:
  \url{http://proceedings.mlr.press/v70/tandon17a.html}
\BIBentrySTDinterwordspacing

\bibitem{NIPS2017_7027}
\BIBentryALTinterwordspacing
Q.~Yu, M.~Maddah-Ali, and S.~Avestimehr, ``Polynomial codes: an optimal design
  for high-dimensional coded matrix multiplication,'' in \emph{Advances in
  Neural Information Processing Systems 30}, I.~Guyon, U.~V. Luxburg,
  S.~Bengio, H.~Wallach, R.~Fergus, S.~Vishwanathan, and R.~Garnett, Eds.\hskip
  1em plus 0.5em minus 0.4em\relax Curran Associates, Inc., 2017, pp.
  4403--4413. [Online]. Available:
  \url{http://papers.nips.cc/paper/7027-polynomial-codes-an-optimal-design-for-high-dimensional-coded-matrix-multiplication.pdf}
\BIBentrySTDinterwordspacing

\bibitem{8765375}
S.~{Dutta}, M.~{Fahim}, F.~{Haddadpour}, H.~{Jeong}, V.~{Cadambe}, and
  P.~{Grover}, ``On the optimal recovery threshold of coded matrix
  multiplication,'' \emph{IEEE Transactions on Information Theory}, vol.~66,
  no.~1, pp. 278--301, Jan 2020.

\bibitem{Dutta2018AUC}
S.~Dutta, Z.~Bai, H.~Jeong, T.~M. Low, and P.~Grover, ``A unified coded deep
  neural network training strategy based on generalized polydot codes,''
  \emph{2018 IEEE International Symposium on Information Theory (ISIT)}, pp.
  1585--1589, 2018.

\bibitem{8437563}
Q.~{Yu}, M.~A. {Maddah-Ali}, and A.~S. {Avestimehr}, ``Straggler mitigation in
  distributed matrix multiplication: Fundamental limits and optimal coding,''
  in \emph{2018 IEEE International Symposium on Information Theory (ISIT)},
  June 2018, pp. 2022--2026.

\bibitem{DBLP:journals/corr/abs-1806-00939}
\BIBentryALTinterwordspacing
Q.~Yu, N.~Raviv, J.~So, and A.~S. Avestimehr, ``Lagrange coded computing:
  Optimal design for resiliency, security and privacy,'' \emph{CoRR}, vol.
  abs/1806.00939, 2018. [Online]. Available:
  \url{http://arxiv.org/abs/1806.00939}
\BIBentrySTDinterwordspacing

\bibitem{doi:10.1137/0108018}
\BIBentryALTinterwordspacing
I.~S. Reed and G.~Solomon, ``Polynomial codes over certain finite fields,''
  \emph{Journal of the Society for Industrial and Applied Mathematics}, vol.~8,
  no.~2, pp. 300--304, 1960. [Online]. Available:
  \url{https://doi.org/10.1137/0108018}
\BIBentrySTDinterwordspacing

\bibitem{Hasircioglu2020BivariatePC}
B.~Hasircioglu, J.~G{\'o}mez-Vilardeb{\'o}, and D.~G{\"u}nd{\"u}z, ``Bivariate
  polynomial coding for exploiting stragglers in heterogeneous coded computing
  systems,'' \emph{ArXiv}, vol. abs/2001.07227, 2020.

\bibitem{2018arXiv180600469C}
W.-T. {Chang} and R.~{Tandon}, ``{On the capacity of secure distributed matrix
  multiplication},'' \emph{arXiv e-prints}, Jun. 2018.

\bibitem{Shamir:1979:SS:359168.359176}
\BIBentryALTinterwordspacing
A.~Shamir, ``How to share a secret,'' \emph{Commun. ACM}, vol.~22, no.~11, pp.
  612--613, Nov. 1979. [Online]. Available:
  \url{http://doi.acm.org/10.1145/359168.359176}
\BIBentrySTDinterwordspacing

\bibitem{8849446}
R.~{D'Oliveira}, S.~{El Rouayheb}, and D.~{Karpuk}, ``{GASP codes for secure
  distributed matrix multiplication},'' \emph{arXiv e-prints}, Dec. 2018.

\bibitem{2019arXiv190806957J}
Z.~{Jia} and S.~A. {Jafar}, ``{On the Capacity of Secure Distributed Matrix
  Multiplication},'' \emph{arXiv e-prints}, Aug. 2019.

\bibitem{8341744}
K.~{Banawan} and S.~{Ulukus}, ``Multi-message private information retrieval:
  Capacity results and near-optimal schemes,'' \emph{IEEE Transactions on
  Information Theory}, vol.~64, no.~10, pp. 6842--6862, Oct 2018.

\bibitem{8598994}
R.~{Tajeddine}, O.~W. {Gnilke}, D.~{Karpuk}, R.~{Freij-Hollanti}, and
  C.~{Hollanti}, ``Private information retrieval from coded storage systems
  with colluding, {B}yzantine, and unresponsive servers,'' \emph{IEEE
  Transactions on Information Theory}, vol.~65, no.~6, pp. 3898--3906, June
  2019.

\bibitem{DBLP:journals/corr/abs-2001-05101}
\BIBentryALTinterwordspacing
Q.~Yu and A.~S. Avestimehr, ``Entangled polynomial codes for secure, private,
  and batch distributed matrix multiplication: Breaking the "cubic" barrier,''
  \emph{CoRR}, vol. abs/2001.05101, 2020. [Online]. Available:
  \url{https://arxiv.org/abs/2001.05101}
\BIBentrySTDinterwordspacing

\bibitem{Chen_2020}
\BIBentryALTinterwordspacing
Z.~Chen, Z.~Jia, Z.~Wang, and S.~A. Jafar, ``{GCSA} codes with noise alignment
  for secure coded multi-party batch matrix multiplication,'' in \emph{2020
  {IEEE} International Symposium on Information Theory ({ISIT})}.\hskip 1em
  plus 0.5em minus 0.4em\relax {IEEE}, jun 2020. [Online]. Available:
  \url{https://doi.org/10.1109%2Fisit44484.2020.9174230}
\BIBentrySTDinterwordspacing

\bibitem{Aliasgari:ISIT:19}
M.~{Aliasgari}, O.~{Simeone}, and J.~{Kliewer}, ``Distributed and private coded
  matrix computation with flexible communication load,'' in \emph{2019 IEEE
  International Symposium on Information Theory (ISIT)}, July 2019, pp.
  1092--1096.

\bibitem{2018arXiv180110292D}
S.~{Dutta}, M.~{Fahim}, F.~{Haddadpour}, H.~{Jeong}, V.~{Cadambe}, and
  P.~{Grover}, ``{On the optimal recovery threshold of coded matrix
  multiplication},'' \emph{arXiv e-prints}, Jan. 2018.

\bibitem{Kakar2019UplinkDownlinkTI}
J.~Kakar, A.~Khristoforov, S.~Ebadifar, and A.~Sezgin, ``Uplink-downlink
  tradeoff in secure distributed matrix multiplication,'' \emph{ArXiv}, vol.
  abs/1910.13849, 2019.

\bibitem{8437651}
H.~A. {Nodehi} and M.~A. {Maddah-Ali}, ``Limited-sharing multi-party
  computation for massive matrix operations,'' in \emph{2018 IEEE International
  Symposium on Information Theory (ISIT)}, June 2018, pp. 1231--1235.

\bibitem{10.1145/3335741.3335756}
\BIBentryALTinterwordspacing
M.~Ben-Or, S.~Goldwasser, and A.~Wigderson, \emph{Completeness Theorems for
  Non-Cryptographic Fault-Tolerant Distributed Computation}.\hskip 1em plus
  0.5em minus 0.4em\relax New York, NY, USA: Association for Computing
  Machinery, 2019, p. 351–371. [Online]. Available:
  \url{https://doi.org/10.1145/3335741.3335756}
\BIBentrySTDinterwordspacing

\bibitem{7565465}
S.~{Lin}, T.~Y. {Al-Naffouri}, Y.~S. {Han}, and W.~{Chung}, ``Novel polynomial
  basis with fast {F}ourier transform and its application to {R}eed-{S}olomon
  erasure codes,'' \emph{IEEE Transactions on Information Theory}, vol.~62,
  no.~11, pp. 6284--6299, 2016.

\bibitem{Kakar2018RateEfficiencyAS}
J.~Kakar, S.~Ebadifar, and A.~Sezgin, ``Rate-efficiency and
  straggler-robustness through partition in distributed two-sided secure matrix
  computation,'' \emph{ArXiv}, vol. abs/1810.13006, 2018.

\bibitem{BORODIN1974366}
\BIBentryALTinterwordspacing
A.~Borodin and R.~Moenck, ``Fast modular transforms,'' \emph{Journal of
  Computer and System Sciences}, vol.~8, no.~3, pp. 366 -- 386, 1974. [Online].
  Available:
  \url{http://www.sciencedirect.com/science/article/pii/S0022000074800292}
\BIBentrySTDinterwordspacing

\bibitem{1328724}
M.~{Ylinen}, A.~{Burian}, and J.~{Takala}, ``Direct versus iterative methods
  for fixed-point implementation of matrix inversion,'' in \emph{IEEE Int'l
  Symposium on Circuits and Systems}, vol.~3, 2004, pp. III--225.

\end{thebibliography}
\bibliographystyle{IEEEtran}

\end{document}